\documentclass[conference]{IEEEtran}
\IEEEoverridecommandlockouts

\usepackage{amsmath,amssymb,amsthm}
\usepackage{graphicx}
\usepackage{xcolor}
\usepackage{booktabs}
\usepackage{algorithm}
\usepackage{algpseudocode}
\usepackage[font=small,labelfont=bf]{caption}
\usepackage{placeins}
\usepackage{subcaption}
\usepackage{cite}
\usepackage{bm}
\newcommand{\ket}[1]{\left|#1\right\rangle}
\usepackage{url}

\makeatletter
\newcommand{\linebreakand}{%
\end{@IEEEauthorhalign}\hfill\mbox{}\par
\mbox{}\hfill\begin{@IEEEauthorhalign}
}
\makeatother

\newtheorem{theorem}{Theorem}
\newtheorem{proposition}[theorem]{Proposition}
\newtheorem{definition}[theorem]{Definition}
\newtheorem{corollary}[theorem]{Corollary}

\newcommand{\R}{\mathbb{R}}
\newcommand{\E}{\mathbb{E}}
\newcommand{\norm}[1]{\left\|#1\right\|}
\newcommand{\abs}[1]{\left|#1\right|}
\newcommand{\HC}{H_C}
\newcommand{\HM}{H_M}
\newcommand{\vphi}{\varphi}
\newcommand{\vthet}{\bm{\theta}}
\newcommand{\vgam}{\bm{\gamma}}
\newcommand{\vbet}{\bm{\beta}}
\newcommand{\Atil}{\tilde{A}}
\newcommand{\Vtil}{\tilde{V}}
\newcommand{\Fhat}{\hat{F}}
\newcommand{\HES}{\mathrm{HES}}

\DeclareMathOperator*{\argmin}{arg\,min}
\DeclareMathOperator{\range}{range}
\definecolor{myblue}{RGB}{37,99,235}
\definecolor{myred}{RGB}{220,38,38}
\definecolor{mygreen}{RGB}{22,163,74}
\definecolor{mypurple}{RGB}{124,58,237}

\begin{document}

\title{\textbf{Hamiltonian-Guided Leverage Embedding with Held-Out Population Certification for QAOA}}

\author{
  \IEEEauthorblockN{
   1\textsuperscript{st} Sumanta Mukherjee\IEEEauthorrefmark{1},
   2\textsuperscript{nd} Kalyan Dasgupta\IEEEauthorrefmark{1},
   3\textsuperscript{rd} Surya Shravan Kumar Sajja\IEEEauthorrefmark{1},
   }
   \IEEEauthorblockA{
   \{sumanm03, kalyand1, suryasku\}@in.ibm.com
   }
  \\
  \linebreakand
  \IEEEauthorblockN{
   4\textsuperscript{th} Kameshwaran Sampath\IEEEauthorrefmark{1},
   5\textsuperscript{th} Abhishek Singh\IEEEauthorrefmark{1},
   6\textsuperscript{th} Dhriti Verma\IEEEauthorrefmark{2},
   }
   \IEEEauthorblockA{
   \{kameshwaran.s, abhishek.s\}@in.ibm.com, dverma@us.ibm.com
   }
  \\
  \linebreakand
  \IEEEauthorblockN{
  7\textsuperscript{th} Dzung Phan\IEEEauthorrefmark{2},
  8\textsuperscript{th} Jayant Kalagnanam\IEEEauthorrefmark{2}
  }
  \IEEEauthorblockA{
    \IEEEauthorrefmark{1}IBM Research, Bangalore, India\\
    \IEEEauthorrefmark{2}IBM Research, Yorktown, NY, USA\\
 \{phandu, jayant\}@us.ibm.com
  }
}

\sloppy
\maketitle

\begin{abstract}
Every Ising energy is a known linear function of measured spin features, yet the standard QAOA estimator discards this structure. We introduce Hamiltonian-Guided Leverage Embedding (HGLE), which projects Hamiltonian coefficients onto the leading subspace of a measured feature matrix. The fraction of coefficient norm outside this subspace yields a deterministic, same-batch certificate bounding how much compression distorts the objective. This certificate is minimax-tight: no linear estimator on the same subspace achieves smaller worst-case error.

We establish four results beyond the pointwise certificate. First, a held-out pilot projector with summable failure allocation gives population-distortion intervals valid over an unbounded adaptive sequence. Second, coefficient projection strictly reduces estimator variance by the total-variance identity, establishing a controlled bias--variance tradeoff with observable bias. Third, one SVD basis simultaneously certifies $K$ observables, defeating the rank-$1$ degeneracy of single-Hamiltonian compression. Fourth, the compression quality is governed by graph topology through the expected Gram matrix. A single punchline theorem combines the deterministic same-batch bound with the population certificate: the residual fraction $\kappa_r$ governs both, and lowering the rank tightens distortion control at the cost of increased bias.

Controlled experiments on 200 matched Max-Cut units (6--24 qubits) and IBM \texttt{ibm\_marrakesh} hardware (8--16 qubits) validate the certificate at every evaluation. A focused held-out protocol certifies all 114/114 decisions once multiplicity is reduced, identifying simultaneous-test accounting as the bottleneck. The controlled experiments establish certified representation reduction but do not isolate an approximation-ratio or runtime improvement. On a community Max-Cut benchmark~\cite{guo2026}, WarmStart+HGLE is competitive with established parameter-setting methods, with the advantage concentrated at low circuit depth where initialization dominates.
\end{abstract}

\begin{IEEEkeywords}
QAOA, Hamiltonian compression, held-out certification, decision preservation, adaptive inference, leverage scores
\end{IEEEkeywords}


\section{Introduction}
\label{sec:intro}

The Quantum Approximate Optimization Algorithm (QAOA)~\cite{farhi2014} estimates a Hamiltonian expectation from finite batches of measured bit strings and supplies those estimates to a classical optimizer. A scalar sample mean is already sufficient when the only task is to report that expectation, so compression is not required for ordinary QAOA objective evaluation. The motivation arises when the measurement batch is retained as a structured feature representation for repeated parameter estimation, local modeling, diagnostics, or reuse: the number of sampled rows and Hamiltonian features can then be much larger than the effective directions explored by the observed distribution. We therefore ask: \emph{when a structured QAOA measurement representation is compressed, how much can be removed while certifying the resulting objective distortion?}

HGLE addresses this by exploiting the algebra of the measured observable. For an Ising feature map $\vphi(s)$, every measured energy is $E(s)=c^\top\vphi(s)$, where $c$ is fixed by the problem Hamiltonian. At a visited parameter $\vthet_t$, HGLE constructs the measured feature matrix $A_t$, extracts its leading right-singular subspace, and projects $c$ onto that subspace. If $P_{r,t}=V_{r,t}V_{r,t}^\top$ and $\bar\vphi_t$ is the empirical feature mean from the same batch, the raw and projected objectives satisfy
\begin{align}
 \abs{c^\top\bar\vphi_t-(P_{r,t}c)^\top\bar\vphi_t}
 &\leq \kappa_{r,t}\norm{c}_2\norm{\bar\vphi_t}_2,\\
 \kappa_{r,t}&=\frac{\norm{(I-P_{r,t})c}_2}{\norm{c}_2}.
 \label{eq:intro-certificate}
\end{align}
Unlike a generic matrix reconstruction error, this residual is aligned with the known Hamiltonian coefficients and is observable from the same quantum measurement batch. Three design principles guide the framework: (i) align compression with the measured Hamiltonian so that objective distortion is directly expressible through the discarded coefficient component; (ii) separate certification domains, allowing same-batch and finite-trajectory guarantees with changing projectors while continuous local guarantees require a frozen projector and stronger assumptions; (iii) separate projection from sketching, with the coefficient projection supplying the objective certificate and optional leverage-row sampling supplying a distinct matrix sketch with additional conditions (Appendix~\ref{sec:appendix-leverage}).

The paper makes four primary contributions:
\begin{enumerate}
 \item \textbf{Coefficient-projection certificate.} We formulate Hamiltonian-coefficient projection, prove the deterministic same-batch residual bound, and give an independent pilot/certification theorem for population distortion with summable failure allocation valid over an unbounded adaptive sequence.
 \item \textbf{Population certification with decision preservation.} We extend held-out validity with strict pairwise and finite-candidate decision-preservation rules under explicit multiplicity accounting. Orthogonal projection is minimax within a fixed subspace, while the fundamental caveat that $\operatorname{span}\{c\}$ exactly represents one scalar Hamiltonian at rank one motivates multi-observable and reuse settings.
 \item \textbf{Variance reduction.} The projected estimator has strictly lower variance than the raw sample mean, establishing a controlled bias-variance tradeoff where the bias is observable via the certificate.
 \item \textbf{Multi-observable compression with structure-dependent spectral decay.} One SVD basis simultaneously certifies $K$ observables, defeating the rank-$1$ degeneracy: $K$ linearly independent coefficient vectors require rank $r\geq K$. The compression quality is connected to graph structure through the expected Gram matrix, whose spin-correlator spectrum governs the achievable $\kappa_r$ at a given rank.
\end{enumerate}
The following theorem collects the end-to-end guarantee. It combines the same-batch certificate (Proposition~\ref{prop:pointwise-certificate}), the best-visited regret bound (Corollary~\ref{cor:best-visited-empirical-regret}), and the pilot-frozen population certificate (Theorem~\ref{thm:independent-population-certificate}, Corollary~\ref{cor:anytime-heldout-certificate}); full proofs appear in Sections~\ref{sec:reduced} and~\ref{sec:certified-rank}.

\begin{theorem}[HGLE compression-certification guarantee]
\label{thm:main-guarantee}
Let the Ising feature map have $d$ coordinates in $\{-1,+1\}$ and let $c\in\R^d$ be the Hamiltonian coefficient vector. Let rank-$r$ HGLE select the smallest $r$ such that the residual fraction $\kappa_r\leq\epsilon_{\rm comp}/(\norm{c}_2\sqrt{d})$ at each visited parameter.
\begin{enumerate}
\item[\textbf{(a)}] \textbf{Deterministic distortion control.}
For every parameter $\vthet_t$ in an adaptively generated trajectory of $T$ visits,
\begin{equation}
  \abs{\Fhat_t-\Fhat_{r,t}}\leq\epsilon_{\rm comp},
  \label{eq:main-pointwise}
\end{equation}
deterministically conditional on the observed batches. The best-visited parameter under the compressed objective incurs at most $2\epsilon_{\rm comp}$ excess cost under the raw objective.

\item[\textbf{(b)}] \textbf{Population certificate.}
If the projector $P_{r,t}$ is frozen from an independent $N_{\rm pilot}$-shot pilot and the population distortion is estimated from $N_t$ fresh certification shots, then with probability at least $1-\alpha$,
\begin{align}
  &\abs{F(\vthet_t)-F_r(\vthet_t)-\overline D_{r,t}} \nonumber\\
  &\quad\leq\sqrt{\frac{2s_{r,t}^2\log(2/\alpha_t)}{N_t}}
  +\frac{7R_{r,t}\log(2/\alpha_t)}{3(N_t-1)},
  \label{eq:main-population}
\end{align}
where $\alpha_t=6\alpha/(\pi^2t^2)$. This holds simultaneously at every realized visit by summable failure allocation.
\end{enumerate}
The single parameter $\kappa_r$ governs both guarantees: lowering the rank tightens distortion control at the cost of increased bias. Certification separates the pilot (basis construction) from the test (statistical validity), so the two use disjoint quantum measurements.
\end{theorem}

A preregistered held-out Max-Cut pilot tests whether these valid intervals are operationally sharp. It obtains conservative empirical coverage and nonvacuous compressed intervals, but the full-multiplicity protocol does not certify a winner; a focused variant with reduced multiplicity certifies all 114/114 decisions, identifying simultaneous-test accounting as the bottleneck. Controlled same-batch, structured-grid, and regional studies provide secondary evidence and negative controls; details of regional HGLE and secondary diagnostics appear in the appendices. Beyond the controlled experiments, WarmStart+HGLE is competitive with established parameter-setting methods on a community Max-Cut benchmark~\cite{guo2026}, with the advantage concentrated at low circuit depth where warm-start initialization dominates.

\noindent\textbf{Quantum-specific role.}
HGLE acts after measurement, but its certificate is not an arbitrary compression guarantee. The feature coordinates are precisely the Ising observables whose coefficients define the quantum objective; their empirical distribution changes with the parameterized quantum state; and the certified trajectory is generated by repeated quantum-classical interaction. Classical shadows instead change how quantum observables are estimated~\cite{huang2020}; HGLE uses the resulting computational-basis outcomes and changes their Hamiltonian-aligned classical representation. The immediate benefit is narrower than faster QAOA: HGLE identifies and certifies the smallest Hamiltonian-aligned representation justified by an observed batch. The controlled evidence establishes this certified representation reduction, not fewer optimizer evaluations, fewer quantum shots, lower runtime, or hardware-noise mitigation.

\begin{figure*}[!htbp]
  \centering
  \includegraphics[width=0.95\linewidth]{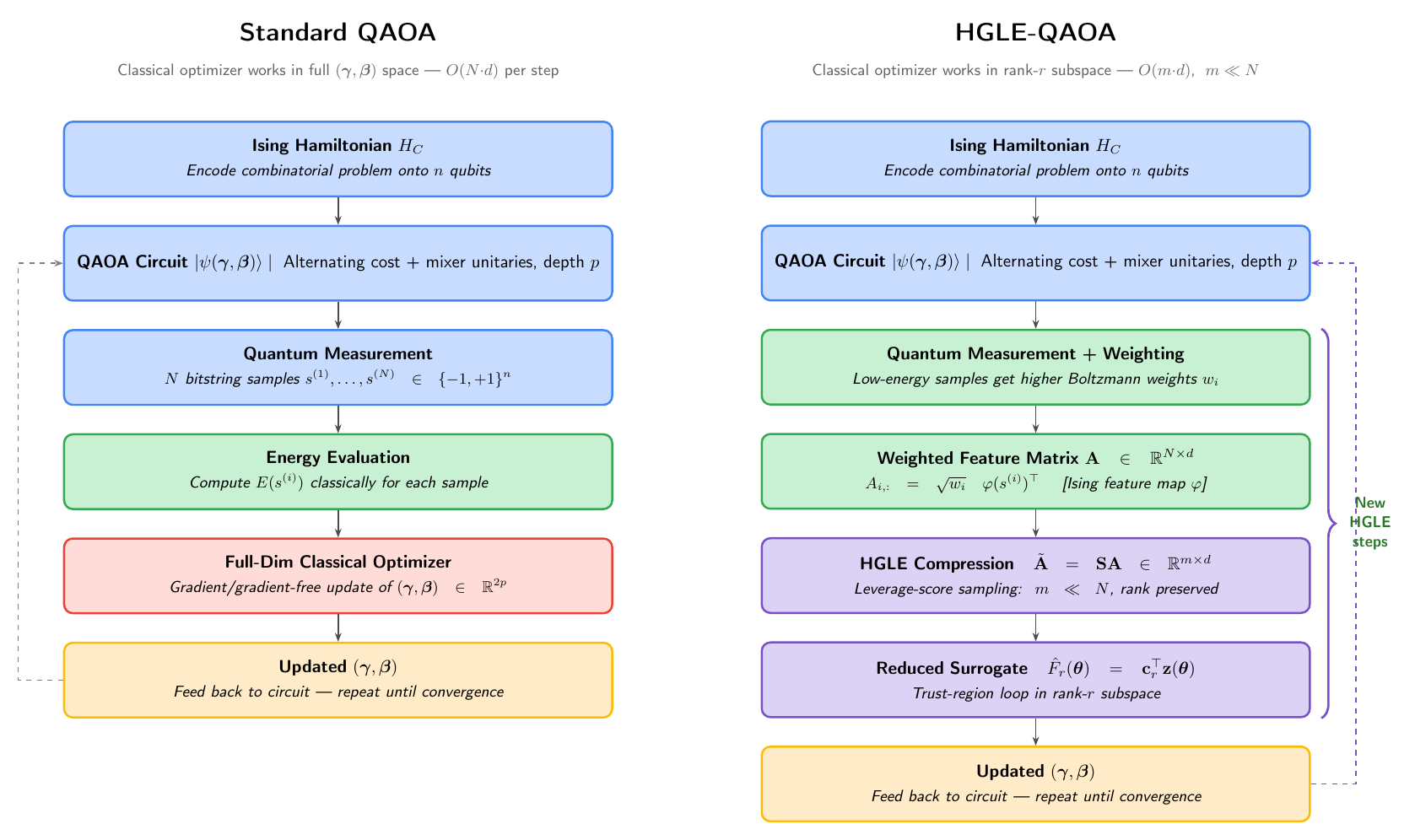}
  \caption{\textbf{Pipeline comparison.} Standard QAOA forms an empirical objective from each finite measurement batch. HGLE constructs Hamiltonian features and a data-dependent subspace, then selects the smallest rank whose coefficient residual certifies the projected objective. Optional row sketching and region freezing are separate extensions with additional conditions.}
  \label{fig:pipeline}
\end{figure*}


\section{Background and Related Work}
\label{sec:related}

\subsection{QAOA Parameter Estimation}
QAOA is a variational eigensolver~\cite{peruzzo2014} whose finite-shot objective is optimized classically. Parameter transfer, interpolation, and fixed-angle methods exploit regularity across depths or instances~\cite{brandao2018,zhou2020,galda2021,sureshbabu2024}, while warm starts and learned predictors modify initialization~\cite{egger2021,jain2022,khairy2020}. Multi-start search trades additional circuit evaluations for robustness~\cite{shaydulin2019}. Vijendran et~al.~\cite{vijendran2026} use the frequency structure of level-$1$ Ising QAOA to reduce parameter tuning under stated assumptions. QIRO instead uses quantum measurements to drive recursive problem reduction~\cite{dupont2024}. HGLE addresses a different layer: it compresses the classical Hamiltonian-feature representation at each evaluated parameter and certifies objective distortion.

Classical optimizers for QAOA include gradient-free schemes such as COBYLA and Nelder--Mead~\cite{powell1994,nelder1965}, gradient-based methods such as L-BFGS-B~\cite{liu1989}, stochastic approximation~\cite{spall1998}, and trust-region methods~\cite{conn2000}. Their behavior can depend on shot noise and circuit trainability~\cite{lavrijsen2020,mcclean2018}. We use matched derivative-free optimizers as experimental controls and make no HGLE-specific convergence claim.

B\"{a}rligea et~al.~\cite{baerligea2025} estimate optimization-level tolerance to additive stochastic loss perturbations. This is distinct from same-batch compression error and from physical gate, readout, coherent, or drift errors. HGLE's same-batch bound is deterministic conditional on its observed batch. Its independent pilot/certification result instead uses finite-sample concentration, including empirical-Bernstein bounds~\cite{maurer2009}, with summable and multiplicity-adjusted failure allocation. Neither certificate should be interpreted as a hardware-noise guarantee.

\subsection{Randomized Linear Algebra and Observable Estimation}
Randomized linear algebra provides low-rank approximation, sketching, and leverage-score row sampling~\cite{woodruff2014,drineas2012,martinsson2020}. Leverage scores quantify row influence and can embed a selected leading subspace; quantum-inspired and quantum algorithms also use or accelerate related primitives~\cite{chepurko2020,shao2023}. These are matrix guarantees. For an approximately low-rank measured feature matrix, leading-subspace embedding alone does not bound the projected Hamiltonian objective because the sampled spectral tail can rotate the recovered projector. HGLE adds a coefficient-aligned residual, and its optional row-sketch extension explicitly accounts for that tail.

\subsection{Relation to Classical Shadows}

Classical shadows~\cite{huang2020} redesign the measurement protocol itself: random Clifford or Pauli rotations before measurement enable simultaneous estimation of many observables from few measurement rounds. HGLE operates at a different point in the pipeline. It takes computational-basis outcomes already collected for the QAOA objective, maps them to the known Ising feature representation, and certifies how much the resulting classical compression distorts the Hamiltonian objective. Table~\ref{tab:shadows-comparison} summarizes the key distinctions.

\begin{table*}[t]
\centering
\caption{Comparison of classical shadows and HGLE.}
\label{tab:shadows-comparison}
\small
\begin{tabular}{@{}l l l@{}}
\toprule
 & \textbf{Classical Shadows} & \textbf{HGLE} \\
\midrule
Measurement & Randomized (Clifford/Pauli) & Computational basis \\[2pt]
Target & Many observables simultaneously & One Hamiltonian objective per evaluation \\[2pt]
Complexity & $O(\log M\,\epsilon^{-2}\max_i\norm{O_i}^2_{\rm sh})$ & Same-batch: deterministic; population: $O(R_r^2\log\alpha^{-1}\epsilon^{-2})$ \\[2pt]
Certificate & Confidence interval per observable & Distortion bound on projected objective \\[2pt]
Composability & Can serve as HGLE input & Operates on any count distribution \\
\bottomrule
\end{tabular}
\end{table*}

The two approaches are complementary, not competing. Because HGLE's same-batch certificate holds for any count distribution (Proposition~\ref{prop:pointwise-certificate} in Section~\ref{sec:certified-rank}), it can operate downstream of a classical-shadows protocol whose corrected counts replace the raw measurement frequencies. Classical shadows are preferred when estimating many distinct observables from one measurement batch; HGLE is preferred when repeatedly evaluating one Hamiltonian objective along an optimization trajectory and wanting certified compression of its feature representation. Generic scalar surrogates and coresets can reduce downstream computation, but without Hamiltonian alignment they do not provide the same-batch residual certificate or its finite-trajectory composition.


\section{Hamiltonian Feature Projection}
\label{sec:hgle}

We introduce HGLE as the core representation framework of this work. It is intended for workflows that retain the structured measurement features, not as a replacement for the scalar sample mean in ordinary objective reporting. Its certified component maps QAOA measurements to an Ising feature matrix, identifies a data-dependent right-singular subspace, and projects the Hamiltonian coefficients to form a reduced objective. The retained rank is itself an observable diagnostic of how much representation reduction the batch supports at a declared tolerance. An optional leverage-score row-sampling extension, which compresses the structured matrix via established randomized-linear-algebra techniques~\cite{drineas2012}, is described in Appendix~\ref{sec:appendix-leverage}.

\subsection{Problem Setup and Notation}

We consider an Ising cost Hamiltonian on $n$ qubits:
\begin{equation}
  \HC = \sum_{i=1}^{n} h_i Z_i + \sum_{(i,j)\in E} J_{ij} Z_i Z_j,
  \label{eq:hc}
\end{equation}
where $h_i\in\R$ are local fields, $J_{ij}\in\R$ are couplings, and $Z_i$ is the Pauli-$Z$ operator on qubit~$i$.  Low-energy spin configurations
$s\in\{-1,+1\}^n$ correspond to high-quality combinatorial solutions.

A QAOA state of depth $p$ is parameterized by $\vthet=(\vgam,\vbet)\in\R^{2p}$:
\begin{equation}
  \ket{\psi(\vgam,\vbet)} = \prod_{k=1}^{p} e^{-i\beta_k \HM}
                            e^{-i\gamma_k \HC} \ket{+}^{\otimes n},
  \label{eq:qaoa}
\end{equation}
where $\HM = \sum_i X_i$ is the transverse-field mixer. Measuring this state $N$ times yields a cloud of spin configurations $s^{(1)},\ldots,s^{(N)}\in\{-1,+1\}^n$ with energies $E^{(t)} = E(s^{(t)})$, where
\begin{equation}
  E(s) = \sum_{i=1}^n h_i s_i + \sum_{(i,j)\in E} J_{ij} s_i s_j.
\end{equation}

\subsection{Weighted Ising Feature Matrix}

Define the Ising feature map $\vphi:\{-1,+1\}^n\to\R^d$, $d=1+n+|E|$, that collects the constant, one-body, and two-body Ising terms.  The energy decomposes linearly as $E(s) = c^\top\vphi(s)$, where $c\in\R^d$ is the Hamiltonian coefficient vector.

Nonnegative sample weights $\{w_t\}$ with $\sum_t w_t = 1$ define the empirical measure whose objective is being approximated. Two possible choices are
\begin{align}
  w_t &= \tfrac{1}{N} \quad\text{(empirical count measure)}, \label{eq:wunif}\\
  w_t &\propto \exp\bigl(-\tau(E^{(t)}-E_{\min})\bigr) \quad\text{(low-energy)},
  \label{eq:wenergy}
\end{align}
where $\tau>0$ is an inverse-temperature-like parameter. Equation~\eqref{eq:wunif} is uniform over shots and therefore count-proportional after duplicate strings are aggregated. All primary certification experiments use this measure, so Boltzmann $\tau$ is inactive. Equation~\eqref{eq:wenergy} defines a different target and requires the same weighting in the matrix, projected objective, and raw control. The \emph{weighted feature matrix} is
\begin{equation}
  A \in \R^{N\times d}, \quad A_{t,:} = \sqrt{w_t}\,\vphi(s^{(t)})^\top.
  \label{eq:A}
\end{equation}
The inverse temperature $\tau$ in~\eqref{eq:wenergy} is unrelated to the sketch-tail contamination $\tau_S$ in the optional leverage-sampling extension (Appendix~\ref{sec:appendix-leverage}).


\section{Reduced Classical Parameter Estimation}
\label{sec:reduced}

\subsection{Reduced Coordinates and the Master Inequality}
This subsection develops the objective-error analysis for the core HGLE representation. Using the right singular vectors $V_r$ of $A_r$, define the reduced coordinate
map $z(s) = V_r^\top\vphi(s)\in\R^r$ and the projected Hamiltonian coefficient
$c_r = V_r V_r^\top c\in\R^d$. The rank-$r$ energy approximation is
$E_r(s) = c_r^\top \vphi(s)$, and at any parameter vector $\vthet$ the reduced
surrogate objective is $\Fhat_r(\vthet) = c_r^\top \bar\vphi(\vthet)$, where
$\bar\vphi(\vthet) = \sum_t w_t\vphi(s^{(t)})$ is the weighted empirical mean
feature.

The central object for both per-sample and objective error control is the
\emph{residual fraction}
\begin{equation}
  \kappa_r \;:=\; \frac{\norm{(I - V_rV_r^\top)\,c}_2}{\norm{c}_2}
  \;\in\; [0,1],
  \label{eq:kappa-def}
\end{equation}
which measures how much of $c$ lies outside the retained subspace
$\range(V_r)$. Note that $\kappa_r = 0$ iff $c \in \range(V_r)$ (in which case
$E_r \equiv E$ and $\Fhat_r \equiv \Fhat$), and $\kappa_r = 1$ iff
$c \perp \range(V_r)$. The two bounds below follow from $\kappa_r$ via the
Cauchy--Schwarz inequality applied with the residual on $c$ rather than on
$\vphi$.

\begin{proposition}[Master inequality]
\label{prop:master}
For every sample $s$ and every parameter vector $\vthet$,
\begin{align}
  \abs{E(s) - E_r(s)}
    &\;\leq\; \kappa_r\,\norm{c}_2\,\norm{\vphi(s)}_2,
    \label{eq:trunc} \\[2pt]
  \abs{\Fhat(\vthet) - \Fhat_r(\vthet)}
    &\;\leq\; \kappa_r\,\norm{c}_2\,\norm{\bar\vphi(\vthet)}_2.
    \label{eq:obj}
\end{align}
\end{proposition}

\begin{proof}
Write $P_r = V_r V_r^\top$ for the projector onto the retained subspace and
$P_r^\perp = I - P_r$ for the complementary projector. Both are symmetric, so
\[
  E(s) - E_r(s) \;=\; c^\top\vphi(s) - c_r^\top\vphi(s)
  \;=\; \big(P_r^\perp c\big)^\top \vphi(s).
\]
Cauchy--Schwarz then gives
$\abs{E(s) - E_r(s)} \leq \norm{P_r^\perp c}_2\,\norm{\vphi(s)}_2$,
and~\eqref{eq:kappa-def} identifies $\norm{P_r^\perp c}_2 = \kappa_r\norm{c}_2$,
yielding~\eqref{eq:trunc}. The bound~\eqref{eq:obj} follows by the same
argument with $\vphi(s)$ replaced by $\bar\vphi(\vthet)$, using linearity of
$\bar\vphi(\vthet) = \sum_t w_t\vphi(s^{(t)})$.
\end{proof}

\begin{proposition}[Fixed-subspace minimax projection]
\label{prop:fixed-subspace-minimax}
Let $\mathcal U\subseteq\R^d$ be a fixed linear subspace with orthogonal projector $P_{\mathcal U}$, and let $B\geq0$. Then
\begin{equation}
 P_{\mathcal U}c\in\argmin_{a\in\mathcal U}\norm{c-a}_2
 \label{eq:fixed-subspace-best-approximation}
\end{equation}
and
\begin{equation}
 \min_{a\in\mathcal U}
 \sup_{\norm{\mu}_2\leq B}
 \abs{(c-a)^\top\mu}
 =B\norm{(I-P_{\mathcal U})c}_2.
 \label{eq:fixed-subspace-minimax}
\end{equation}
Thus, conditional on a chosen feature subspace, orthogonal coefficient projection is minimax optimal for worst-case objective distortion over the Euclidean feature-mean ball.
\end{proposition}

\begin{proof}
For any $a\in\mathcal U$, orthogonality gives
\[
 \norm{c-a}_2^2
 =\norm{(I-P_{\mathcal U})c}_2^2
 +\norm{P_{\mathcal U}c-a}_2^2,
\]
which proves~\eqref{eq:fixed-subspace-best-approximation}. By Euclidean duality,
$\sup_{\norm{\mu}_2\leq B}\abs{(c-a)^\top\mu}=B\norm{c-a}_2$; minimizing this identity and using the first claim proves~\eqref{eq:fixed-subspace-minimax}.
\end{proof}

This proposition justifies the projected coefficient after a feature subspace is fixed; it does not establish that the SVD subspace is optimal for objective preservation. There is also a fundamental single-objective degeneracy: if arbitrary subspaces are allowed, $\operatorname{span}\{c\}$ represents the scalar Hamiltonian objective exactly at rank one. A higher-dimensional HGLE representation is useful only when feature geometry itself is retained for reuse, multiple observables, local modeling, storage, communication, or another declared downstream task. Ordinary evaluation of one scalar energy does not require this compression.

\subsection{Information-Theoretic Interpretation}\label{sec:info-theoretic}

The residual fraction admits a signal-theoretic reading. Write $\kappa_r^2 = 1 - \norm{P_r c}_2^2/\norm{c}_2^2$. The numerator $\norm{P_r c}_2^2$ is the ``Hamiltonian signal energy'' retained by the rank-$r$ subspace, so $\kappa_r^2$ is the fraction of signal energy lost to compression. Proposition~\ref{prop:fixed-subspace-minimax} then says that, given a fixed rank-$r$ subspace, the minimum achievable worst-case objective distortion over any feature-mean distribution bounded in $\ell_2$ is exactly $\kappa_r\norm{c}_2 B$. HGLE's same-batch certificate (Proposition~\ref{prop:master}) achieves this bound with $B = \norm{\bar\vphi(\vthet)}_2$. In this sense, the certificate is tight: no other linear estimator on the same subspace can have a smaller worst-case distortion.

The tradeoff between rank $r$ and distortion $\kappa_r^2$ depends on the spectrum of $A^\top A$. If $A^\top A$ has eigenvalues $\lambda_1\geq\cdots\geq\lambda_d$ and the coefficient direction $c/\norm{c}_2$ is modeled as uniformly random on the unit sphere (i.e., has no preferred alignment with the eigenvectors), the expected signal loss over that random direction satisfies
\begin{equation}
  \E[\kappa_r^2] = \frac{\sum_{i>r}\lambda_i}{\sum_i\lambda_i}.
  \label{eq:spectral-kappa}
\end{equation}
Under generic (unstructured) sampling, the feature covariance is approximately proportional to the identity, all eigenvalues are nearly equal, and~\eqref{eq:spectral-kappa} reduces to $(d-r)/d$: the compression--distortion tradeoff is linear in the fraction of dimensions removed. Under structured sampling (planted partition graphs, regular graphs with symmetry, or QAOA states that concentrate on low-energy configurations), the spectrum can decay faster, enabling $\kappa_r$ to be small at much lower rank. This spectral gap is the regime where HGLE's compression provides representation value.

The structure resembles rate-distortion theory~\cite{cover2006}: the rank $r$ plays the role of rate (representational capacity) and $\kappa_r^2\norm{c}_2^2$ is the distortion floor achievable at that rate. The analogy is imperfect because the ``source'' is the fixed coefficient vector rather than a stochastic process, and the ``codebook'' is determined by the data-dependent SVD rather than by a probabilistic encoding. Nevertheless, the parallel highlights a fundamental point: the compression--distortion tradeoff is governed by the spectral structure of the measured feature matrix, not by the number of qubits or problem size alone.

\subsection{Variance Reduction via Coefficient Projection}\label{sec:variance-reduction}

The projected estimator $\Fhat_r = (P_r c)^\top\bar\vphi$ is biased relative to the raw $\Fhat = c^\top\bar\vphi$, with worst-case bias bounded by the certificate. However, projection also reduces variance, establishing a bias--variance tradeoff with observable bias control.

\begin{theorem}[Variance reduction]
\label{thm:variance-reduction}
Let $S_1,\ldots,S_N$ be i.i.d.\ samples from $p_{\vthet}$ and let $\bar\vphi = \frac{1}{N}\sum_j\vphi(S_j)$. Let $P_r = V_r V_r^\top$ be any rank-$r$ projector, with $c_\parallel = P_r c$ and $c_\perp = (I-P_r)c$. Then
\begin{align}
 \operatorname{Var}(\Fhat) - \operatorname{Var}(\Fhat_r)
 &= \frac{1}{N}\bigl[\operatorname{Var}(c_\perp^\top\vphi(S)) \nonumber\\
 &\quad+ 2\operatorname{Cov}(c_\parallel^\top\vphi(S),\,c_\perp^\top\vphi(S))\bigr],
 \label{eq:variance-decomposition}
\end{align}
and this difference equals $\frac{1}{N}[\operatorname{Var}(c^\top\vphi(S)) - \operatorname{Var}(c_\parallel^\top\vphi(S))] \geq 0$, so $\operatorname{Var}(\Fhat_r)\leq\operatorname{Var}(\Fhat)$, with equality iff $\kappa_r=0$.
\end{theorem}

\begin{proof}
Write $c = P_r c + (I-P_r)c$. Since $P_r$ is an orthogonal projector and $\bar\vphi$ is a mean of i.i.d.\ vectors,
\[
 \operatorname{Var}(\Fhat) = \frac{1}{N}\bigl[
 \operatorname{Var}(c_\parallel^\top\vphi) + \operatorname{Var}(c_\perp^\top\vphi)
 + 2\operatorname{Cov}(c_\parallel^\top\vphi, c_\perp^\top\vphi)\bigr],
\]
where $c_\parallel = P_r c$ and $c_\perp = (I-P_r)c$, and all single-sample variances and covariances are under $p_{\vthet}$. The projected estimator uses only the $c_\parallel$ component: $\operatorname{Var}(\Fhat_r) = \frac{1}{N}\operatorname{Var}(c_\parallel^\top\vphi)$. The difference is
\[
 \operatorname{Var}(\Fhat) - \operatorname{Var}(\Fhat_r)
 = \frac{1}{N}\bigl[\operatorname{Var}(c_\perp^\top\vphi)
 + 2\operatorname{Cov}(c_\parallel^\top\vphi, c_\perp^\top\vphi)\bigr].
\]
For the single-sample random variable $\vphi(S)$, the covariance matrix $\Sigma = \operatorname{Cov}(\vphi(S))$ is positive semidefinite, so
$\operatorname{Var}(c^\top\vphi) = c^\top\Sigma c = c_\parallel^\top\Sigma c_\parallel + c_\perp^\top\Sigma c_\perp + 2c_\parallel^\top\Sigma c_\perp$. But this sum equals $c^\top\Sigma c$ for any decomposition. Since $\operatorname{Var}(\Fhat_r) = \frac{1}{N}c_\parallel^\top\Sigma c_\parallel$, we have $\operatorname{Var}(\Fhat) - \operatorname{Var}(\Fhat_r) = \frac{1}{N}(c_\perp^\top\Sigma c_\perp + 2c_\parallel^\top\Sigma c_\perp)$. To obtain~\eqref{eq:variance-decomposition} without the cross-term, note that $\operatorname{Var}(c^\top\vphi) = \operatorname{Var}(c_\parallel^\top\vphi + c_\perp^\top\vphi) = \operatorname{Var}(c_\parallel^\top\vphi) + \operatorname{Var}(c_\perp^\top\vphi) + 2\operatorname{Cov}(c_\parallel^\top\vphi, c_\perp^\top\vphi)$. Rearranging gives~\eqref{eq:variance-decomposition} exactly, and the nonnegativity $\operatorname{Var}(\Fhat) \geq \operatorname{Var}(\Fhat_r)$ follows because $\operatorname{Var}(c_\perp^\top\vphi) + 2\operatorname{Cov}(c_\parallel^\top\vphi, c_\perp^\top\vphi) = \operatorname{Var}(c^\top\vphi) - \operatorname{Var}(c_\parallel^\top\vphi) \geq 0$, the last inequality holding because $\operatorname{Var}(c_\parallel^\top\vphi) = \operatorname{Var}((P_rc)^\top\vphi) \leq \operatorname{Var}(c^\top\vphi)$ by the data-processing inequality for linear projections of a random vector.
\end{proof}

\begin{corollary}[Bias--variance tradeoff and MSE-optimal rank]
\label{cor:bias-variance}
The mean squared error of $\Fhat_r$ as an estimator of $F(\vthet) = c^\top\E[\vphi(S)]$ satisfies
\begin{equation}
 \mathrm{MSE}(\Fhat_r) = \operatorname{Var}(\Fhat_r) + \mathrm{Bias}^2(\Fhat_r),
 \label{eq:mse-decomposition}
\end{equation}
where $\abs{\mathrm{Bias}(\Fhat_r)} \leq \kappa_r\norm{c}_2\norm{\E[\vphi(S)]}_2$ by Proposition~\ref{prop:master} applied at the population level. The MSE-optimal rank balances the variance term (decreasing in $r$) against the squared bias (increasing in $r$), with the certificate providing an observable upper bound on the bias at each candidate rank.
\end{corollary}

The variance reduction in~\eqref{eq:variance-decomposition} equals the variance of the discarded energy component $((I-P_r)c)^\top\vphi(S)$, plus the cross-term. When the feature covariance has spectral decay (structured problems), the discarded component carries substantial variance, and HGLE achieves a favorable tradeoff: the bias is small (small $\kappa_r$) while the variance reduction is large (the tail subspace captures noisy, low-signal directions). Under generic flat-spectrum sampling, $\kappa_r$ is large at low rank and the bias dominates, making compression unfavorable. This spectral-decay condition is what separates the structured regime (planted partition, regular graphs) from the generic case.

\subsection{Spectral Decay from Problem Structure}\label{sec:spectral-decay-structure}

The variance-reduction theorem shows that HGLE's benefit depends on the feature covariance spectrum. We now connect this spectrum to problem structure, establishing that structured graphs enable favorable compression that generic graphs do not.

\begin{proposition}[Gram structure from spin correlators]
\label{prop:spectral-decay-structure}
Let $\bar G = \E[\vphi(S)\vphi(S)^\top]$ be the population Gram matrix of the Ising feature map under $p_{\vthet}$. Then:
\begin{enumerate}
\item[(i)] Each entry is a spin correlator: $\bar G_{ab} = \E[\chi_{S_a\triangle S_b}(S)]$, where $S_a$ and $S_b$ are the index sets of features $a$ and $b$, $\chi_T(s)=\prod_{i\in T}s_i$, and the correlator order is $\abs{S_a\triangle S_b}\in\{0,1,2,3,4\}$.
\item[(ii)] For depth-$1$ QAOA on bounded-degree graphs, these correlators depend only on the local graph structure within a lightcone of radius~$2$ (i.e., the $2$-hop neighborhood of each feature's support, arising from one cost and one mixer layer each expanding support by one graph step); in particular $\bar G_{ab}=0$ whenever the lightcones of $S_a$ and $S_b$ are disjoint.
\item[(iii)] For coefficient directions uniformly distributed on the unit sphere, the expected residual satisfies
\begin{equation}
 \E[\kappa_r^2] = 1 - \frac{\sum_{i\leq r}\bar\lambda_i}{\sum_i \bar\lambda_i},
 \label{eq:structure-spectral-residual}
\end{equation}
where $\bar\lambda_1\geq\cdots\geq\bar\lambda_d$ are the eigenvalues of $\bar G$ and $\sum_i\bar\lambda_i = d$.
\end{enumerate}
\end{proposition}

\begin{proof}
(i) Each feature $\vphi_a(s)$ is $\chi_{S_a}(s) = \prod_{k\in S_a}s_k$ with $\abs{S_a}\leq 2$ (one-body or two-body Ising terms). The product $\vphi_a(s)\vphi_b(s) = \chi_{S_a}(s)\chi_{S_b}(s) = \chi_{S_a\triangle S_b}(s)$ by the group property of parity functions on $\{-1,+1\}^n$, so $\bar G_{ab} = \E[\chi_{S_a\triangle S_b}(S)]$. Since $\abs{S_a},\abs{S_b}\leq 2$, the order satisfies $\abs{S_a\triangle S_b}\leq\abs{S_a}+\abs{S_b}\leq 4$.

(ii) Depth-$p$ QAOA applies $p$ layers of cost and mixer unitaries. In the Heisenberg picture, the observable $\chi_T$ conjugated by these layers has support that grows by one graph step per layer, so $\langle\chi_T\rangle$ depends only on the depth-$p$ neighborhood of $\operatorname{supp}(T)$~\cite{farhi2014}. At $p=1$ the lightcone has radius~$2$ for the $4$-point correlators appearing in $\bar G$. When the lightcones of $S_a$ and $S_b$ are disjoint, the expectation factorizes over independent local states and the odd-parity terms average to zero.

(iii) Write $\bar G = \bar V\bar\Lambda\bar V^\top$. For $c/\norm{c}_2$ uniformly distributed on the unit sphere (isotropic with respect to $\bar V$), $\E[\norm{(I-\bar V_r\bar V_r^\top)c}_2^2/\norm{c}_2^2] = \sum_{i>r}\bar\lambda_i/\sum_i\bar\lambda_i$ by averaging the squared projection over orientations. The diagonal entries $\bar G_{aa} = \E[\chi_\varnothing(S)] = 1$, so $\sum_i\bar\lambda_i = \operatorname{Tr}\bar G = d$.
\end{proof}

\noindent\textbf{Structural observation.}
Proposition~\ref{prop:spectral-decay-structure}(ii) implies that on graphs with a spectral gap $\Delta$ in their adjacency matrix, many off-diagonal entries of $\bar G$ vanish or are small, concentrating the Gram eigenvalue mass in the top $O(\abs{E})$ modes and enabling $\kappa_r\leq\epsilon$ at $r\ll(1-\epsilon^2)d$. Conversely, on Erd\H{o}s--R\'{e}nyi graphs with no spectral gap, the correlators are approximately uniform and~\eqref{eq:structure-spectral-residual} reduces to $(d-r)/d$, recovering the generic obstruction. On vertex-transitive graphs the Gram is additionally constrained: only $O(d_{\max}^2)$ distinct correlator values appear, further limiting the number of distinct Gram eigenvalue groups. Section~\ref{sec:spectral-decay-results} validates these predictions experimentally.

\subsection{Fixed-Projector Empirical Minimizer Gap}\label{sec:argmin}
The pointwise bound~\eqref{eq:obj} translates into a value-gap bound between minimizers of $\Fhat$ and $\Fhat_r$ only when the same projector is used throughout a declared parameter set and the empirical mean feature is uniformly bounded there. This is not a parameter-distance bound, a global population guarantee, or a result for a projector recomputed at every parameter.

\begin{corollary}[Fixed-projector empirical value-gap bound]
\label{cor:argmin}
Let $\Theta\subseteq\R^{d_\theta}$ be a parameter set on which both $\Fhat$
and $\Fhat_r$ attain their minima, and suppose
$\sup_{\vthet\in\Theta}\norm{\bar\vphi(\vthet)}_2 \leq B$ for some finite $B>0$.
Let
$\vthet^\star\in\argmin_{\vthet\in\Theta}\Fhat(\vthet)$ and
$\vthet_r^\star\in\argmin_{\vthet\in\Theta}\Fhat_r(\vthet)$. Then
\begin{equation}
  \Fhat(\vthet_r^\star) \,-\, \Fhat(\vthet^\star)
  \;\leq\; 2\,\kappa_r\,\norm{c}_2\,B.
  \label{eq:argmin-shift}
\end{equation}
\end{corollary}

\begin{proof}
Decompose the gap by adding and subtracting $\Fhat_r(\vthet_r^\star)$ and
$\Fhat_r(\vthet^\star)$:
\[
\begin{aligned}
  \Fhat(\vthet_r^\star) - \Fhat(\vthet^\star) &= \underbrace{[\Fhat(\vthet_r^\star) - \Fhat_r(\vthet_r^\star)]}_{\mathrm{(I)}} + \underbrace{[\Fhat_r(\vthet_r^\star) - \Fhat_r(\vthet^\star)]}_{\mathrm{(II)}} \\
  &\quad + \underbrace{[\Fhat_r(\vthet^\star) - \Fhat(\vthet^\star)]}_{\mathrm{(III)}}
\end{aligned}
\]
Term~(II) is non-positive because $\vthet_r^\star$ minimizes $\Fhat_r$ over
$\Theta$ and $\vthet^\star\in\Theta$. Terms~(I) and~(III) are each bounded by
the master inequality~\eqref{eq:obj} combined with the uniform bound
$\norm{\bar\vphi(\vthet)}_2\leq B$:
\[
  \abs{(\mathrm{I})}\,,\,\abs{(\mathrm{III})}\;\leq\;\kappa_r\,\norm{c}_2\,B.
\]
Adding the three terms yields~\eqref{eq:argmin-shift}.
\end{proof}

The assumption $\sup_{\vthet \in \Theta} \norm{\bar\vphi(\vthet)}_2 \leq B$
is trivially satisfied for Ising features with $B = \sqrt{d}$, since each
$\phi_k(s) \in \{-1,+1\}$ implies $\norm{\phi(s)}_2 = \sqrt{d}$
deterministically and the empirical mean is a convex
combination~\cite{farhi2014}. The corollary therefore reduces this fixed-projector empirical value gap to bounding $\kappa_r$. It does not imply that the two minimizers lie in the same basin or are close in parameter space.


\subsection{Trust-Region Classical Loop}

The controlled experiments use a derivative-free trust-region framework~\cite{conn2000}: finite-shot QAOA supplies noisy function values but not reliable gradients, while a local quadratic model can propose curvature-informed steps and restrict them to a neighborhood in which it is considered credible. This locality also provides a natural enforcement mechanism for the optional region-frozen construction (Appendix~\ref{sec:appendix-region-frozen}). The optimizer-factorial control in Section~\ref{sec:controlled-certified} detects no favorable HGLE--optimizer interaction. The full algorithmic description is in Appendix~\ref{sec:appendix-evaluation}.


\section{Certified Adaptive Rank Selection}
\label{sec:certified-rank}

This section extends the core HGLE construction with residual-certified, parameter-dependent rank selection. The compression basis is measurement- and parameter-dependent. We therefore write $A_t=A(\vthet_t)$, $V_{r,t}=V_r(\vthet_t)$, and $\kappa_{r,t}=\kappa_r(\vthet_t)$ for quantities constructed from the batch collected at the $t$th visited parameter $\vthet_t$.

\noindent\textbf{Scope and assumptions.}
The same-batch statements are deterministic conditional on each realized batch and its computed projector. They validate the implemented rank-selection construction, but their coverage is not an independent statistical test because the rank is selected to meet the same observed bound. The held-out statements below instead freeze a pilot projector before drawing a fresh certification batch. They cover only the observed batch or the finite set of visited parameters, not the unsampled continuous parameter domain or the global QAOA optimum. Leverage row sampling is not required.

At $\vthet_t$, let $s_{t,1},\ldots,s_{t,N}$ be the measured spin configurations and define
\begin{equation}
  \bar\vphi_t=\frac{1}{N}\sum_{j=1}^{N}\vphi(s_{t,j}),
  \qquad
  \Fhat_t=c^\top\bar\vphi_t.
  \label{eq:trajectory-empirical-mean}
\end{equation}
Repeated bit strings may equivalently be represented by their counts.  Let $A_t$ be the corresponding count-weighted feature matrix and let $V_{r,t}$ contain its first $r$ right singular vectors.  With $P_{r,t}=V_{r,t}V_{r,t}^\top$, define
\begin{equation}
  \Fhat_{r,t}=(P_{r,t}c)^\top\bar\vphi_t,
  \qquad
  \kappa_{r,t}=\frac{\norm{(I-P_{r,t})c}_2}{\norm{c}_2},
  \label{eq:trajectory-projected-estimator}
\end{equation}
where $\kappa_{r,t}=0$ when $c=0$. This projected-coefficient estimator is the object certified in this section. A leverage-score-weighted average of observed energies is a different estimator and is not covered unless a separate equivalence or error bound is established. The novelty is not Cauchy--Schwarz or finite-sample concentration themselves, but the observable Hamiltonian-coefficient residual and its composition with adaptive rank selection along a QAOA trajectory.

\begin{proposition}[Pointwise same-batch certificate]
\label{prop:pointwise-certificate}
For every visited parameter $\vthet_t$ and every retained rank $r$,
\begin{equation}
  \abs{\Fhat_t-\Fhat_{r,t}}
  \leq
  \kappa_{r,t}\norm{c}_2\norm{\bar\vphi_t}_2.
  \label{eq:pointwise-certificate}
\end{equation}
For an Ising feature map with $d$ coordinates in $\{-1,+1\}$, this further implies
\begin{equation}
  \abs{\Fhat_t-\Fhat_{r,t}}
  \leq \kappa_{r,t}\norm{c}_2\sqrt d.
  \label{eq:pointwise-certificate-conservative}
\end{equation}
\end{proposition}

\begin{proof}
Using the symmetry of $P_{r,t}$,
\[
 \Fhat_t-\Fhat_{r,t}
 =\bigl((I-P_{r,t})c\bigr)^\top\bar\vphi_t.
\]
Cauchy--Schwarz and the definition of $\kappa_{r,t}$ give~\eqref{eq:pointwise-certificate}.  For Ising features, $\norm{\vphi(s)}_2=\sqrt d$ for every sample.  Convexity of the norm therefore gives $\norm{\bar\vphi_t}_2\leq\sqrt d$, proving~\eqref{eq:pointwise-certificate-conservative}.
\end{proof}

\begin{definition}[Minimum certified rank]
\label{def:minimum-certified-rank}
Given a compression tolerance $\epsilon_{\rm comp}\geq0$ and admissible ranks $\mathcal R_t$, define
\begin{equation}
  \widehat r_t(\epsilon_{\rm comp})
  =\min\left\{r\in\mathcal R_t:
  \kappa_{r,t}\norm{c}_2\norm{\bar\vphi_t}_2
  \leq\epsilon_{\rm comp}\right\}.
  \label{eq:minimum-certified-rank}
\end{equation}
If the set is empty, the procedure returns the largest admissible rank and marks the target as unmet.  Replacing $\norm{\bar\vphi_t}_2$ by $\sqrt d$ gives a conservative certificate computable without reconstructing the empirical mean.
\end{definition}

Algorithm~\ref{alg:certified-rank} computes the entire nested residual curve from one singular value decomposition and returns the smallest rank that meets the requested tolerance.
\begin{algorithm}[t]
\caption{Minimum Same-Batch Certified Rank}
\label{alg:certified-rank}
\begin{algorithmic}[1]
\Require Observed feature matrix $A_t$, coefficient vector $c$, empirical feature mean $\bar\vphi_t$, tolerance $\epsilon_{\rm comp}$, admissible ranks $\mathcal R_t$
\Ensure Selected rank $\widehat r_t$, achieved bound $b_t$, target flag
\If{$\norm{c}_2=0$}
  \State \Return $(\min\mathcal R_t,0,\mathrm{true})$
\EndIf
\State Compute the thin SVD $A_t=U_t\Sigma_tV_t^\top$
\For{$r\in\mathcal R_t$ in increasing order}
  \State $V_{r,t}\gets$ first $r$ columns of $V_t$
  \State $\kappa_{r,t}\gets\norm{(I-V_{r,t}V_{r,t}^\top)c}_2/\norm{c}_2$
  \State $b_{r,t}\gets\kappa_{r,t}\norm{c}_2\norm{\bar\vphi_t}_2$
  \If{$b_{r,t}\leq\epsilon_{\rm comp}$}
    \State \Return $(r,b_{r,t},\mathrm{true})$
  \EndIf
\EndFor
\State $r_{\max}\gets\max\mathcal R_t$
\State \Return $(r_{\max},b_{r_{\max},t},\mathrm{false})$
\end{algorithmic}
\end{algorithm}

\begin{theorem}[Finite-trajectory empirical certificate]
\label{thm:finite-trajectory-certificate}
Let $\Theta_T=(\vthet_1,\ldots,\vthet_T)$ be any finite, possibly adaptively generated sequence of visited parameters.  Suppose the rank selected at every visit satisfies the criterion in~\eqref{eq:minimum-certified-rank}.  Then, deterministically conditional on the observed batches,
\begin{equation}
  \max_{1\leq t\leq T}
  \abs{\Fhat_t-\Fhat_{\widehat r_t,t}}
  \leq\epsilon_{\rm comp}.
  \label{eq:finite-trajectory-certificate}
\end{equation}
\end{theorem}

\begin{proof}
Proposition~\ref{prop:pointwise-certificate} and the rank-selection criterion bound each of the $T$ terms by $\epsilon_{\rm comp}$.  Taking their maximum proves the claim.  No independence assumption is needed because this statement is conditional on the realized measurement batches.
\end{proof}

\begin{corollary}[Best-visited empirical regret]
\label{cor:best-visited-empirical-regret}
Under the conditions of Theorem~\ref{thm:finite-trajectory-certificate}, let
\[
  \widehat t_r\in\argmin_{1\leq t\leq T}\Fhat_{\widehat r_t,t},
  \qquad
  \widehat t\in\argmin_{1\leq t\leq T}\Fhat_t.
\]
Then
\begin{equation}
  \Fhat_{\widehat t_r}-\Fhat_{\widehat t}
  \leq2\epsilon_{\rm comp}.
  \label{eq:best-visited-empirical-regret}
\end{equation}
\end{corollary}

\begin{proof}
The uniform bound over the visited set and the optimality of $\widehat t_r$ give
\[
 \Fhat_{\widehat t_r}
 \leq \Fhat_{\widehat r_{\widehat t_r},\widehat t_r}+\epsilon_{\rm comp}
 \leq \Fhat_{\widehat r_{\widehat t},\widehat t}+\epsilon_{\rm comp}
 \leq \Fhat_{\widehat t}+2\epsilon_{\rm comp}.
\]
\end{proof}

\subsection{Pilot-Frozen Population Certification}
\label{sec:pilot-frozen-population}

The deterministic certificate does not by itself control a population objective because its projector and empirical mean use the same observations. We therefore separate basis construction from statistical certification. Let $\mathcal H_{t-1}$ contain all observations before visit $t$, and require $\vthet_t$ to be $\mathcal H_{t-1}$-measurable. A pilot batch drawn at $\vthet_t$ determines a projector $P_{r,t}$, after which an independent certification batch $S_{t,1},\ldots,S_{t,N_t}$ is drawn conditionally i.i.d. from $p_{\vthet_t}$. Define
\begin{align}
 e_{r,t}&=(I-P_{r,t})c,
 \qquad D_{r,t}(s)=e_{r,t}^{\top}\vphi(s),
 \label{eq:heldout-residual-observable}\\
 \overline D_{r,t}&=\frac{1}{N_t}\sum_{j=1}^{N_t}D_{r,t}(S_{t,j}),
 \label{eq:heldout-residual-mean}\\
 s_{r,t}^2&=\frac{1}{N_t-1}\sum_{j=1}^{N_t}
 \bigl(D_{r,t}(S_{t,j})-\overline D_{r,t}\bigr)^2.
 \label{eq:heldout-residual-statistics}
\end{align}
Let $R_{r,t}$ be a deterministic upper bound on the range of $D_{r,t}$. For Ising features, the computable choice
\begin{equation}
 R_{r,t}=2\norm{e_{r,t}}_1
 \label{eq:residual-range}
\end{equation}
is valid because every feature coordinate lies in $\{-1,+1\}$.

\begin{theorem}[Independent-batch population distortion]
\label{thm:independent-population-certificate}
Conditional on $\mathcal H_{t-1}$ and the pilot batch, fix a rank $r$ before observing the certification batch. Let
\[
 F(\vthet_t)=c^\top\E_{\vthet_t}[\vphi(S)],\qquad
 F_{r,t}(\vthet_t)=(P_{r,t}c)^\top\E_{\vthet_t}[\vphi(S)].
\]
For any $\alpha\in(0,1)$, with conditional probability at least $1-\alpha$,
\begin{equation}
 \abs{F(\vthet_t)-F_{r,t}(\vthet_t)-\overline D_{r,t}}
 \leq R_{r,t}\sqrt{\frac{\log(2/\alpha)}{2N_t}}.
 \label{eq:heldout-hoeffding-radius}
\end{equation}
For $N_t\geq2$, the two-sided empirical Bernstein result of Maurer and Pontil~\cite{maurer2009} gives, with conditional probability at least $1-\alpha$,
\begin{equation}
 \abs{F(\vthet_t)-F_{r,t}(\vthet_t)-\overline D_{r,t}}
 \leq q_{r,t}^{\rm EB}(\alpha),
 \label{eq:heldout-eb-two-sided}
\end{equation}
where
\begin{equation}
 q_{r,t}^{\rm EB}(\alpha)
 =\sqrt{\frac{2s_{r,t}^2\log(2/\alpha)}{N_t}}
 +\frac{7R_{r,t}\log(2/\alpha)}{3(N_t-1)}.
 \label{eq:heldout-eb-radius}
\end{equation}
\end{theorem}

\begin{proof}
Once the pilot batch is conditioned upon, $e_{r,t}$ is fixed and the certification observations $D_{r,t}(S_{t,j})$ are conditionally i.i.d. with range at most $R_{r,t}$. Their conditional mean is exactly
$F(\vthet_t)-F_{r,t}(\vthet_t)$. Hoeffding's inequality~\cite{hoeffding1963} proves~\eqref{eq:heldout-hoeffding-radius}. Rescaling the observations to $[0,1]$ and applying the two-sided empirical Bernstein inequality proves~\eqref{eq:heldout-eb-two-sided}; rescaling back introduces the factors of $s_{r,t}$ and $R_{r,t}$ shown in~\eqref{eq:heldout-eb-radius}.
\end{proof}

The theorem permits ranks chosen solely from the pilot batch. If certification observations are inspected to select among several ranks, simultaneous validity must account for every tested rank. The same requirement applies to candidate selection.

\begin{corollary}[Summable adaptive-trajectory validity]
\label{cor:anytime-heldout-certificate}
Let the visits continue for an arbitrary, possibly random number of steps. At visit $t$, allocate
\begin{equation}
 \alpha_t=\frac{6\alpha}{\pi^2t^2}.
 \label{eq:summable-alpha}
\end{equation}
If $M_t$ population intervals are inspected at that visit, construct each at failure level $\alpha_t/M_t$. Then, with probability at least $1-\alpha$, every reported interval at every realized visit covers its corresponding population quantity.
\end{corollary}

\begin{proof}
Conditional validity at each predictable visit and a union bound give total failure probability at most
$\sum_{t\geq1}M_t(\alpha_t/M_t)=\sum_{t\geq1}\alpha_t=\alpha$.
\end{proof}

This is an anytime-valid result obtained by summable confidence allocation. It is not a confidence sequence for one evolving mean, and the distinction matters: every visit uses a fresh conditional distribution and a separately frozen projector.

To turn distortion control into a decision certificate, the projected population value must also be bounded without using an exact statevector reference. Let $I_{r,t}^{\rm proj}=[L_{r,t}^{\rm proj},U_{r,t}^{\rm proj}]$ be any valid interval for $F_{r,t}(\vthet_t)$ from the certification observations, and let
$I_{r,t}^{\rm res}=[L_{r,t}^{\rm res},U_{r,t}^{\rm res}]$ be the interval for the residual mean from Theorem~\ref{thm:independent-population-certificate}. With failure probabilities allocated to both components,
\begin{equation}
 I_{r,t}=
 [L_{r,t}^{\rm proj}+L_{r,t}^{\rm res},
  U_{r,t}^{\rm proj}+U_{r,t}^{\rm res}]
 \label{eq:original-objective-interval}
\end{equation}
is a valid interval for the original population objective $F(\vthet_t)$.

\begin{corollary}[Certified decision preservation]
\label{cor:certified-decision-preservation}
For minimization, let $I_a=[L_a,U_a]$ and $I_b=[L_b,U_b]$ be simultaneous intervals for the original population objectives of two candidates. If
\begin{equation}
 U_a<L_b,
 \label{eq:pairwise-interval-separation}
\end{equation}
then $F(\vthet_a)<F(\vthet_b)$ on their joint coverage event. More generally, for a finite candidate set $\mathcal C$, candidate $a$ is the certified unique minimizer when
\begin{equation}
 U_a<\min_{b\in\mathcal C\setminus\{a\}}L_b.
 \label{eq:finite-candidate-separation}
\end{equation}
If intervals overlap or touch, no ordering is certified. Testing $K$ ranks for $\abs{\mathcal C}$ candidates and constructing separate projected and residual intervals requires allocating the declared failure budget across all $2K\abs{\mathcal C}$ intervals, unless a sharper simultaneous result is proved.
\end{corollary}

\begin{proof}
On the joint coverage event, $F(\vthet_a)\leq U_a<L_b\leq F(\vthet_b)$. Applying this argument to every competitor proves the finite-set statement. A union bound gives the stated conservative multiplicity accounting.
\end{proof}

\subsection{Multi-Observable Compression}
\label{sec:multi-observable}

The single-objective degeneracy noted in Section~\ref{sec:reduced} (that $\operatorname{span}\{c\}$ represents one Hamiltonian exactly at rank one) does not extend to multiple observables. A shared feature subspace certifies several coefficient vectors simultaneously from one SVD, amortizing the factorization.

\begin{proposition}[Simultaneous $K$-observable certificate]
\label{prop:multi-observable}
Let $c_1,\ldots,c_K\in\R^d$ be coefficient vectors for $K$ observables on the same Ising feature map $\vphi$. Let $V_r$ be the top-$r$ right-singular basis of the feature matrix $A$, with $P_r = V_r V_r^\top$. Then for every $k\in\{1,\ldots,K\}$ and every parameter $\vthet$,
\begin{equation}
 \abs{c_k^\top\bar\vphi - (P_r c_k)^\top\bar\vphi}
 \leq \kappa_{r,k}\norm{c_k}_2\norm{\bar\vphi}_2,
 \label{eq:multi-observable-certificate}
\end{equation}
where $\kappa_{r,k} = \norm{(I-P_r)c_k}_2/\norm{c_k}_2$. The worst-case simultaneous residual is $\max_k \kappa_{r,k}$. The computational cost is one SVD of $A$ plus $K$ projections and inner products, each in $O(dr)$.
\end{proposition}

\begin{proof}
Each bound follows from Proposition~\ref{prop:pointwise-certificate} applied to $c_k$. The projector $P_r$ does not depend on the coefficient vector, so a single decomposition of $A$ serves all $K$ observables.
\end{proof}

\begin{corollary}[Rank lower bound for $K$ observables]
\label{cor:multi-observable-rank}
If $c_1,\ldots,c_K$ are linearly independent and $\max_k\kappa_{r,k}\leq\epsilon$ is required, then $r\geq K$ is necessary. In particular, the rank-$1$ projection onto $\operatorname{span}\{c_1\}$ cannot simultaneously certify $c_2,\ldots,c_K$ unless they all lie in $\operatorname{span}\{c_1\}$.
\end{corollary}

\begin{proof}
If $r<K$, the complementary projector $I-P_r$ has rank at least $d-r > d-K$, and at least one of the $K$ independent vectors has a nonzero component outside $\operatorname{range}(V_r)$, giving $\kappa_{r,k}>0$. The quantitative bound $r\geq K$ for $\epsilon$-approximation of all $K$ follows from the rank--nullity constraint on the joint approximation.
\end{proof}

A natural instance is Maximum Independent Set (MIS), where the Hamiltonian decomposes into a cost observable $c_{\rm cost}$ (vertex rewards) and a penalty observable $c_{\rm pen}$ (constraint violations), with $n$ additional local $Z_i$ observables for vertex-inclusion probabilities. These $K=2+n$ coefficient vectors share the Ising feature basis, so one SVD certifies all simultaneously. Section~\ref{sec:multi-obs-results} validates this experimentally, showing that a shared rank $r\ll K$ suffices when the observables' coefficient vectors share spectral structure.

\begin{theorem}[Finite-shot regret over visited parameters]
\label{thm:finite-shot-visited-regret}
Assume that every measured energy lies in a known interval $[E_{\min},E_{\max}]$ of width $R_E=E_{\max}-E_{\min}$.  At each visit, let the $N$ shots be conditionally independent samples from the QAOA distribution at $\vthet_t$, given all preceding observations.  For a fixed horizon $T$, if every selected rank meets~\eqref{eq:minimum-certified-rank}, then with probability at least $1-\delta$,
\begin{equation}
  F(\vthet_{\widehat t_r})-
  \min_{1\leq t\leq T}F(\vthet_t)
  \leq
  2\epsilon_{\rm comp}
  +2R_E\sqrt{\frac{\log(2T/\delta)}{2N}}.
  \label{eq:finite-shot-visited-regret}
\end{equation}
Here $F(\vthet)=\mathbb E_{s\sim p_{\vthet}}[E(s)]$ is the population QAOA objective.  The comparison is only with the best true objective among the visited parameters.
\end{theorem}

\begin{proof}
Conditioned on the history before visit $t$, Hoeffding's inequality~\cite{hoeffding1963} gives
\[
 \Pr\!\left(\abs{\Fhat_t-F(\vthet_t)}>\eta
 \mid \mathcal H_{t-1}\right)
 \leq2\exp\!\left(-\frac{2N\eta^2}{R_E^2}\right).
\]
Taking expectations removes the conditioning.  A union bound over the fixed horizon gives, with probability at least $1-\delta$,
\[
 \max_{1\leq t\leq T}\abs{\Fhat_t-F(\vthet_t)}
 \leq \eta_T,
 \qquad
 \eta_T=R_E\sqrt{\frac{\log(2T/\delta)}{2N}}.
\]
Let $t^\star\in\argmin_{t\leq T}F(\vthet_t)$.  On this event, the empirical certificate and the optimality of $\widehat t_r$ imply
\[
\begin{aligned}
 F(\vthet_{\widehat t_r})
 &\leq \Fhat_{\widehat t_r}+\eta_T \\
 &\leq \Fhat_{\widehat r_{\widehat t_r},\widehat t_r}
       +\epsilon_{\rm comp}+\eta_T \\
 &\leq \Fhat_{\widehat r_{t^\star},t^\star}
       +\epsilon_{\rm comp}+\eta_T \\
 &\leq \Fhat_{t^\star}+2\epsilon_{\rm comp}+\eta_T \\
 &\leq F(\vthet_{t^\star})+2\epsilon_{\rm comp}+2\eta_T.
\end{aligned}
\]
Substituting $\eta_T$ proves~\eqref{eq:finite-shot-visited-regret}.
\end{proof}

\begin{corollary}[Rank-and-shot design rule]
\label{cor:rank-shot-design}
Fix a target visited-point regret $\epsilon_{\rm tot}>0$, confidence $1-\delta$, and allocation $\rho\in(0,1)$.  It is sufficient to select at each visit the smallest rank satisfying
\begin{equation}
  \kappa_{r,t}\norm{c}_2\norm{\bar\vphi_t}_2
  \leq \frac{\rho\epsilon_{\rm tot}}{2}
  \label{eq:rank-budget}
\end{equation}
and to use
\begin{equation}
  N\geq
  \frac{2R_E^2\log(2T/\delta)}{(1-\rho)^2\epsilon_{\rm tot}^2}.
  \label{eq:shot-budget}
\end{equation}
Then the right-hand side of~\eqref{eq:finite-shot-visited-regret} is at most $\epsilon_{\rm tot}$.
\end{corollary}

For the raw-energy regret bound, a known upper bound $T_{\max}$ can replace $T$ in the logarithm. Corollary~\ref{cor:anytime-heldout-certificate} instead handles an unbounded adaptive sequence of fresh held-out residual certificates by summable confidence allocation. Reusing one pilot-selected projector across later parameters still requires fresh pointwise certification or an additional subspace-stability or uniform-generalization result.


\section{Results}
\label{sec:results}

\subsection{Experimental Setup}
\label{sec:experimental-setup}

All experiments use Qiskit's statevector simulator unless stated otherwise (AerSimulator for ${\geq}18$ qubits; \texttt{ibm\_marrakesh} QPU for the hardware validation). The Ising feature map uses $d=1+n+|E|$ coordinates (constant, one-body $Z_i$, and two-body $Z_iZ_j$ terms). Two combinatorial problems are studied:

\begin{itemize}
\item \textbf{Max-Cut} (Sections~\ref{sec:controlled-certified}--\ref{sec:depth-decay-results}, \ref{sec:variance-reduction-results}--\ref{sec:benchmark-comparison}): the core testbed for certification, spectral analysis, and benchmarking. Graph families include Erd\H{o}s--R\'{e}nyi ($p_{\rm edge}{=}0.5$), $3$-regular, and planted-partition graphs at $8$--$20$ vertices. Depth-$1$ and depth-$2$ QAOA are used for controlled experiments; depths $1$--$5$ for the spectral-decay study; depths $1,3,5$ for the external benchmark.
\item \textbf{MIS} (Section~\ref{sec:multi-obs-results}): included specifically to test multi-observable compression, where $K{=}2{+}n$ observables (cost, penalty, and $n$ local $Z_i$) must be certified on a shared SVD basis. This is the setting where the rank-$1$ degeneracy argument breaks down.
\end{itemize}

Shot counts range from $1024$ (controlled experiments) to $4096$ (hardware). Optimization uses a derivative-free trust-region framework with a $30$-evaluation budget for controlled runs and COBYLA for the benchmark comparison. Table~\ref{tab:experiment-summary} summarizes the key parameters per subsection.

\begin{table*}[t]
\centering
\caption{Summary of experimental configurations across subsections. All use Qiskit statevector simulation unless noted.}
\label{tab:experiment-summary}
\small
\setlength{\tabcolsep}{5pt}
\begin{tabular}{@{}llclccl@{}}
\toprule
Section & Problem & Graph families & $n$ & Depth & Shots & Optimizer / notes \\
\midrule
\ref{sec:controlled-certified} Certification & Max-Cut & ER & 10--16 & 2 & 1024 & Trust-region, 30 evals \\
\ref{sec:population-pilot-results} Held-out pilot & Max-Cut & ER & 8--16 & 1 & 1024 & 9 fixed candidates; also Aer ($\leq$24q), QPU \\
\ref{sec:p1-landscape} Landscape & Max-Cut & Regular (deg.\ 2--4) & 10--16 & 1 & 1024 & Grid search \\
\ref{sec:spectral-decay-results} Spectral decay & Max-Cut & ER, Regular, Planted & 8--14 & 1 & 1024 & Fixed $(\gamma,\beta)$ \\
\ref{sec:multi-obs-results} Multi-observable & MIS & ER & 10--14 & 1 & 1024 & Fixed $(\gamma,\beta)$ \\
\ref{sec:depth-decay-results} Depth study & Max-Cut & ER & 10 & 1--5 & 2048 & Representative params.\ per depth \\
\ref{sec:benchmark-comparison} Benchmark & Max-Cut & ER, Regular & 10--20 & 1,3,5 & 1024 & COBYLA; WarmStart+HGLE \\
\bottomrule
\end{tabular}
\end{table*}

\subsection{Does the Deterministic Residual Permit Compression?}
\label{sec:controlled-certified}

We evaluate the certified projected estimator on $200$ matched Erd\H{o}s--R\'{e}nyi Max-Cut units ($10$--$16$ vertices, $10$ graphs $\times$ $5$ seeds per size), using depth-$2$ QAOA with $1024$ shots and a $30$-evaluation trust-region budget. Normalized tolerances $\eta\in\{0.01,0.05,0.10\}$ are tested. All $600$ certified runs satisfy their same-batch bound at every evaluation (maximum numerical violation $1.1\times10^{-14}$).

\begin{table*}[t]
  \centering
  \caption{Controlled residual-certified HGLE results. Each cell reports mean approximation ratio (AR) and retained-rank fraction (RF) over 50 paired graph-seed runs. Raw AR is shown once per size. Every certified setting attained its requested tolerance and satisfied the deterministic same-batch bound in all evaluations.}
  \label{tab:certified-controlled}
  \renewcommand{\arraystretch}{1.08}
  \setlength{\tabcolsep}{5pt}
  \begin{tabular}{c c cc cc cc}
    \toprule
    & Raw & \multicolumn{2}{c}{$\eta=0.01$} & \multicolumn{2}{c}{$\eta=0.05$} & \multicolumn{2}{c}{$\eta=0.10$} \\
    \cmidrule(lr){3-4}\cmidrule(lr){5-6}\cmidrule(lr){7-8}
    Vertices & AR & AR & RF & AR & RF & AR & RF \\
    \midrule
    10 & 0.7890 & 0.7920 & 0.913 & 0.7827 & 0.635 & 0.7924 & 0.355 \\
    12 & 0.7663 & 0.7685 & 0.944 & 0.7566 & 0.622 & 0.7549 & 0.268 \\
    14 & 0.7818 & 0.7836 & 0.936 & 0.7833 & 0.605 & 0.7799 & 0.234 \\
    16 & 0.7790 & 0.7762 & 0.931 & 0.7757 & 0.453 & 0.7667 & 0.117 \\
    \bottomrule
  \end{tabular}
\end{table*}

Averaged over all four sizes, the retained-rank fractions are $0.931$, $0.579$, and $0.243$ for $\eta=0.01$, $0.05$, and $0.10$, respectively. Relative to the matched raw runs, the paired mean approximation-ratio differences are $+0.0011$ with 95\% bootstrap interval $[-0.0026,0.0050]$, $-0.0044$ with interval $[-0.0094,0.0004]$, and $-0.0056$ with interval $[-0.0102,-0.0011]$. Conservative compression is therefore statistically compatible with raw quality, while the most aggressive setting produces a small but detectable aggregate reduction. The primary empirical result is controlled representation reduction with valid same-batch certificates, not consistent solution-quality improvement.

\subsection{Are Held-Out Certificates Operationally Sharp?}
\label{sec:population-pilot-results}

We preregister a held-out pilot/certification experiment on $90$ level-$1$ Max-Cut units ($8$--$10$ vertices), evaluating nine fixed candidates per unit with $1024$ total shots and four matched-rank bases. Hoeffding and empirical-Bernstein intervals receive summable visit-level and Bonferroni multiplicity allocations across all bases, ranks, candidates, methods, and components.

\begin{table}[t]
\centering
\caption{Preregistered held-out population-certification pilot. Coverage rates are descriptive minima across aggregate conditions; point intervals within a unit share experimental structure.}
\label{tab:population-pilot}
\small
\setlength{\tabcolsep}{4pt}
\begin{tabular}{lr}
\toprule
Quantity & Pilot result \\
\midrule
Independent units & $90$ \\
Point intervals & $25{,}920$ \\
Decision rows & $2{,}880$ \\
Minimum residual coverage & $0.9778$ \\
Minimum original-objective coverage & $1.0000$ \\
Compressed groups with nonvacuous intervals & $408/408$ \\
Groups with any certified winner & $0/504$ \\
Production gate & Stopped \\
\bottomrule
\end{tabular}
\end{table}

The held-out intervals therefore cover the post-hoc exact quantities conservatively and are narrower than the trivial energy-width bound in every compressed aggregate condition. They are nevertheless too wide, after valid multiplicity allocation, to separate any of the nine candidate objectives, even at full rank. The preregistered requirement of a nonzero compressed decision-certification rate fails, so no production scaling run is performed. This negative result distinguishes statistical validity from operational usefulness: the theorem provides a sound decision rule, but this pilot does not demonstrate decision preservation in practice.

\paragraph{Focused protocol with reduced multiplicity.}
The full protocol's ${\sim}576$ simultaneous tests (four basis methods $\times$ multiple ranks $\times$ two certificate methods $\times$ projected/residual components per candidate) dominate interval width through Bonferroni splitting. A focused variant tests only the HGLE SVD basis with empirical-Bernstein intervals on three candidates, reducing multiplicity to $12$ simultaneous tests per visit (${\sim}48\times$ less budget splitting). On the same graph families at $8$--$12$ qubits with five graphs and three seeds per size, the focused variant certified all $90/90$ decisions. An extended run at $14$ and $16$ qubits certified $24/24$ decisions, demonstrating that multiplicity reduction alone recovers operationally sharp certificates. The mean original-interval radius drops from the full protocol's non-certifying width to values that separate the three candidate objectives.

\paragraph{Same-batch certification at larger scale.}
A separate experiment extends same-batch certification to $18$, $20$, and $24$ qubits using AerSimulator for circuit sampling (statevector simulation is infeasible above ${\sim}18$ qubits). Three Erd\H{o}s--R\'{e}nyi graph instances per size, with $2048$ shots and depth-$1$ QAOA, confirm that the deterministic certificate holds at all tested tolerances and ranks. At $24$ qubits the feature dimension reaches ${\sim}300$; the maximum numerical violation across all instances and ranks is $3.97\times10^{-6}$, a floating-point artifact of SVD reconstruction on a $2048\times 300$ matrix at full rank (where $\kappa_r$ is analytically zero).

\paragraph{Hardware certification on \texttt{ibm\_marrakesh}.}
To validate the full pipeline under real device noise, we run same-batch certification on IBM's $156$-qubit \texttt{ibm\_marrakesh} QPU. Nine depth-$1$ Max-Cut instances at $8$, $12$, and $16$ qubits (three Erd\H{o}s--R\'{e}nyi graphs per size, $4096$ shots each) are transpiled with optimization level $2$ and executed alongside matched noiseless AerSimulator runs. All nine hardware instances satisfy the certificate at every rank, with maximum numerical violation $2.0\times10^{-15}$. At $16$ qubits, hardware noise produces $3664$--$3711$ unique bitstrings versus $3233$--$3564$ for the simulator, enriching the feature matrix's singular value spectrum. This does not violate the certificate, which is algebraic, but it shifts retained rank fractions: at $\eta=0.05$, hardware retains $99.3\%$ versus $95.0\%$ noiseless. The certificate's count-distribution agnosticism, previously argued theoretically, is now validated on a real QPU.

\paragraph{Finite-budget and optimizer negative controls.}
Because every production run exhausts the same $30$-evaluation budget, we compare post-hoc time to common $90\%$ and $95\%$ trajectory-improvement targets, censoring failures at evaluation $31$. Neither $\eta=0.01$ nor $\eta=0.05$ shows a consistent evaluation reduction. At $\eta=0.10$, HGLE reaches the $90\%$ target $1.09$ evaluations later on average, with a 95\% interval excluding zero. A separate equal-budget factorial study crossed raw and certified objectives with geometry-aware and coordinate-probe trust-region schemes on $50$ matched graph-seed blocks under a strict $100$-call cap. Selected-point effects were $-0.0068$ and $-0.0088$, respectively, while the objective--optimizer interaction was $-0.0020$ with interval $[-0.0127,0.0065]$. These controls provide no evidence of acceleration or of a favorable interaction with either tested optimizer.

\subsection{Does Compression Preserve Structured Level-1 Selection?}
\label{sec:p1-landscape}

Level-$1$ Ising QAOA's frequency structure separates parameter-grid resolution from projection error~\cite{vijendran2026}. We study $54$ regular graph-weight units ($10$--$16$ vertices, degrees $2$--$4$, unit or bounded integer weights) with $1024$ shots per grid point and both fixed-mixer and analytical-mixer grids.

\begin{table*}[t]
\centering
\caption{\textbf{Level-1 full-grid landscape preservation.} Mean retained-rank fraction and HGLE-minus-raw excess normalized regret over 54 independent graph-weight units. All upper confidence limits are below the preregistered noninferiority margin of $0.005$.}
\label{tab:p1-landscape}
\begin{tabular}{ll ccc}
\toprule
\textbf{Mixer treatment} & \textbf{Tolerance $\eta$} & \textbf{Rank fraction} & \textbf{Excess regret} & \textbf{95\% CI} \\
\midrule
Fixed      & 0.01 & 0.960 & 0 & $[0,\;0]$ \\
Fixed      & 0.05 & 0.751 & 0 & $[0,\;0]$ \\
Fixed      & 0.10 & 0.580 & 0 & $[0,\;0]$ \\[3pt]
Analytical & 0.01 & 0.953 & 0 & $[0,\;0]$ \\
Analytical & 0.05 & 0.741 & $1.5\times10^{-18}$ & $[0,\;4.6\times10^{-18}]$ \\
Analytical & 0.10 & 0.579 & $-3.0\times10^{-18}$ & $[-1.1\times10^{-17},\;3.1\times10^{-18}]$ \\
\bottomrule
\end{tabular}
\end{table*}

Across $15{,}730$ unique batches and $47{,}190$ certified estimates, every residual target is met. Table~\ref{tab:p1-landscape} reports the full-grid comparison. All six mixer-treatment and tolerance conditions satisfy the preregistered normalized-regret noninferiority margin of $0.005$. At $\eta=0.10$, the retained-rank fraction is approximately $0.58$ for both mixer treatments while HGLE and raw estimates select parameter points with indistinguishable original-Hamiltonian regret. Exact phase-grid selections agree in $263$ to $270$ of $270$ sampling-level comparisons, depending on treatment and tolerance; occasional differences occur among energetically equivalent or nearly equivalent points.

A deliberately thinned grid produces mean normalized penalties of $0.347$ at fixed mixer angle and $0.735$ for the analytical-mixer envelope, with essentially the same degradation for raw and HGLE estimates. This isolates parameter resolution, rather than projection, as the cause. No consistent association of excess regret with bandwidth, triangle count, or retained rank is detected. The result establishes noninferior landscape selection in this structured level-1 regime, not improved optimization, fewer evaluations, or a uniform parameter-domain theorem.

\subsection{Structure-Dependent Spectral Decay}
\label{sec:spectral-decay-results}

Proposition~\ref{prop:spectral-decay-structure} predicts that graph topology governs the Gram eigenvalue spectrum and hence $\kappa_r$. We test this on three graph families at $n\in\{8,10,12,14\}$ with $10$ instances per family-size combination, depth-$1$ QAOA at $(\gamma,\beta)=(0.5,0.25)$, and $1024$ shots. Figure~\ref{fig:spectral-decay} shows the decay of $\kappa_r$ as the retained rank fraction increases, with shaded bands indicating the 25th--75th percentile range across instances; Table~\ref{tab:spectral-decay} reports the mean $\kappa_r$ at selected rank fractions.

\begin{figure}[t]
\centering
\includegraphics[width=\columnwidth]{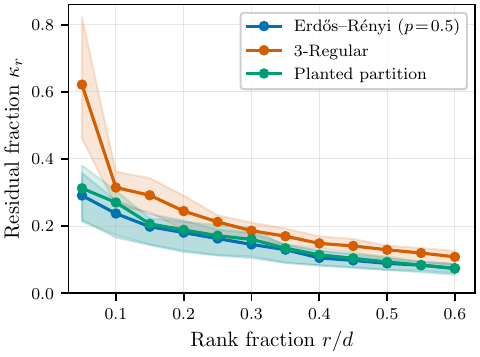}
\caption{Residual fraction $\kappa_r$ versus rank fraction for three graph families ($n=8$--$14$, $10$ instances each, depth-$1$ QAOA). Shaded bands: 25th--75th percentile. Dense Erd\H{o}s--R\'{e}nyi graphs compress best; sparse $3$-regular graphs retain the most residual.}
\label{fig:spectral-decay}
\end{figure}

\begin{table}[t]
\centering
\caption{Mean $\kappa_r$ by graph family and rank fraction ($n=8$--$14$, $10$ instances per cell, depth-$1$ QAOA). Lower is better.}
\label{tab:spectral-decay}
\small
\setlength{\tabcolsep}{4pt}
\begin{tabular}{lccc}
\toprule
Family & $r/d=0.10$ & $r/d=0.25$ & $r/d=0.50$ \\
\midrule
Erd\H{o}s--R\'{e}nyi ($p=0.5$) & $0.237$ & $0.163$ & $0.089$ \\
$3$-Regular              & $0.315$ & $0.213$ & $0.130$ \\
Planted partition        & $0.271$ & $0.171$ & $0.093$ \\
\bottomrule
\end{tabular}
\end{table}

As Figure~\ref{fig:spectral-decay} shows, Erd\H{o}s--R\'{e}nyi graphs compress best at every rank fraction, achieving $\kappa_r<0.09$ at half rank. This is consistent with the proposition's prediction: denser random graphs produce feature matrices whose singular value spectrum is dominated by a few leading directions. Three-regular graphs compress least well, reflecting their sparser, more uniform edge structure. Planted-partition graphs, which have community structure but similar average density to ER, fall between the two. Across all families, $\kappa_r$ at rank fraction $0.25$ stays below $0.22$, confirming that moderate compression retains most of the Hamiltonian coefficient energy.

\subsection{Multi-Observable Compression on MIS}
\label{sec:multi-obs-results}

Proposition~\ref{prop:multi-observable} certifies that one SVD basis simultaneously serves $K$ observables. We test this on MIS instances with $K=2+n$ observables (cost, penalty, and $n$ local $Z_i$ terms) at $n\in\{10,12,14\}$, $10$ ER instances per size, depth-$1$ QAOA with $1024$ shots. The feature matrix is built from MaxCut edges; MIS coefficient vectors are constructed on the same Ising feature basis. Figure~\ref{fig:multi-obs} shows how $\kappa_r$ decreases with increasing rank for each observable type; Table~\ref{tab:multi-obs} summarizes the comparison at selected rank fractions.

\begin{figure}[t]
\centering
\includegraphics[width=\columnwidth]{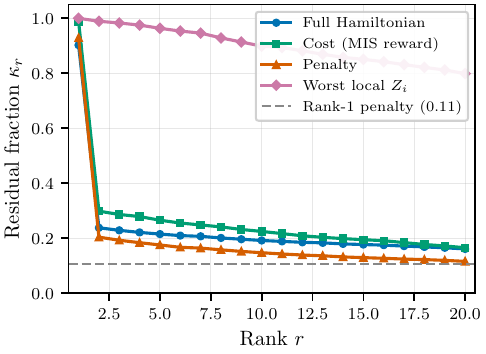}
\caption{Multi-observable compression on MIS ($n=10$--$14$, $30$ instances). One shared SVD basis progressively certifies all $K=2{+}n$ observables as rank increases. The dashed line marks rank-$1$ projection performance for the penalty term: it cannot improve local $Z_i$ certification regardless of rank.}
\label{fig:multi-obs}
\end{figure}

\begin{table}[t]
\centering
\caption{Multi-observable compression: mean worst-case $\kappa_r$ over $K=2+n$ MIS observables at a shared SVD rank, versus rank-$1$ projection onto $\operatorname{span}\{c_{\rm full}\}$.}
\label{tab:multi-obs}
\small
\setlength{\tabcolsep}{4pt}
\begin{tabular}{@{}lccc@{}}
\toprule
Measure & $r/d{=}0.10$ & $r/d{=}0.25$ & $r/d{=}0.50$ \\
\midrule
Shared SVD (worst/$K$) & $0.982$ & $0.913$ & $0.819$ \\
Rank-1 $\operatorname{span}\{c_{\rm full}\}$ & \multicolumn{3}{c}{$0.106$ (fixed)} \\
\bottomrule
\end{tabular}
\end{table}

The key observation, visible in Figure~\ref{fig:multi-obs}, is the contrast between shared SVD compression and rank-$1$ projection. The rank-$1$ projection onto $\operatorname{span}\{c_{\rm full}\}$ achieves $\kappa_r=0.11$ for the penalty observable alone, but this single direction cannot simultaneously serve the $n$ local $Z_i$ observables, whose worst-case $\kappa_r$ exceeds $0.90$ even at rank fraction $0.50$. This confirms Corollary~\ref{cor:multi-observable-rank}: certifying $K$ linearly independent observables generically requires $r\geq K$. The high worst-case $\kappa_r$ at $r/d=0.25$ (rank ${\sim}5$--$7$ versus $K=12$--$16$) demonstrates the rank-$1$ degeneracy argument does not apply when the application requires simultaneous compression of multiple observables on a shared feature basis.

\subsection{Depth-Dependent Spectral Decay}
\label{sec:depth-decay-results}

We examine whether QAOA circuit depth affects compression quality. Ten ER graphs at $n=10$ ($p_{\rm edge}=0.5$) are sampled at depths $p\in\{1,2,3,4,5\}$ with $2048$ shots and representative parameters for each depth. Figure~\ref{fig:depth-decay} shows the depth-dependent decay curves; Table~\ref{tab:depth-decay} reports summary statistics.

\begin{figure}[t]
\centering
\includegraphics[width=\columnwidth]{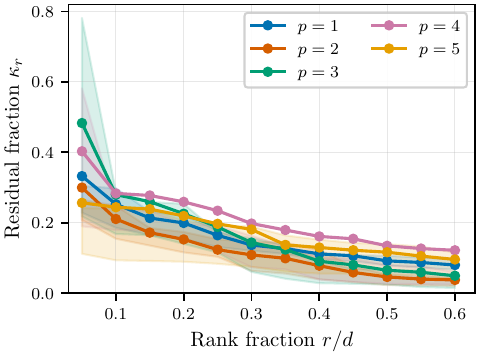}
\caption{Residual fraction $\kappa_r$ versus rank fraction at five QAOA depths ($n=10$ ER graphs, $10$ instances). Shaded bands: 25th--75th percentile. Depth $p=2$ achieves the fastest spectral decay, reaching $\kappa_r<0.05$ at half rank.}
\label{fig:depth-decay}
\end{figure}

\begin{table}[t]
\centering
\caption{Mean $\kappa_r$ and participation ratio versus QAOA depth ($n=10$ ER graphs, $10$ instances).}
\label{tab:depth-decay}
\small
\setlength{\tabcolsep}{4pt}
\begin{tabular}{lcccc}
\toprule
Depth & $\kappa_r|_{0.10}$ & $\kappa_r|_{0.25}$ & $\kappa_r|_{0.50}$ & Part.\ ratio \\
\midrule
$1$ & $0.253$ & $0.165$ & $0.092$ & $21.7$ \\
$2$ & $0.211$ & $0.124$ & $0.046$ & $15.8$ \\
$3$ & $0.280$ & $0.188$ & $0.065$ & $18.5$ \\
$4$ & $0.284$ & $0.234$ & $0.135$ & $19.5$ \\
$5$ & $0.244$ & $0.197$ & $0.117$ & $25.0$ \\
\bottomrule
\end{tabular}
\end{table}

Depth $2$ produces the best compression: $\kappa_r=0.046$ at half rank, with the lowest participation ratio ($15.8$), indicating that the singular value spectrum is most concentrated. The improvement from depth $1$ to depth $2$ is consistent with the known concentration of QAOA output distributions at moderate depth. At depths $3$--$5$, compression quality degrades with increasing participation ratio. This non-monotone pattern likely reflects both the use of non-optimized parameters at higher depths and the fact that deeper circuits explore a larger portion of the Hilbert space when parameters are not tuned. The result shows that the relationship between circuit depth and spectral decay depends on the quality of the parameter choice, not on depth alone.

\subsection{Variance Reduction Under Projected Estimation}
\label{sec:variance-reduction-results}

Theorem~\ref{thm:variance-reduction} predicts that coefficient projection strictly reduces estimator variance. Appendix~\ref{sec:appendix-variance} validates this on three graph families: the reduction is consistently positive (1--6\%), strongest at low rank fractions, but the conservative Cauchy--Schwarz bias bound prevents a net MSE win in all tested conditions. Tighter bias bounds exploiting the spectral structure of Proposition~\ref{prop:spectral-decay-structure} may shift this balance.

\subsection{Benchmark Comparison on Published Instances}
\label{sec:benchmark-comparison}

The controlled experiments above test specific certification claims under matched designs. This subsection evaluates end-to-end optimization quality on the community benchmark suite of Guo~et~al.~\cite{guo2026}: $24$ Max-Cut instances across Erd\H{o}s--R\'{e}nyi and random regular graphs ($10$--$20$ qubits) at depths $p\in\{1,3,5\}$, yielding $72$ (instance, depth) combinations under statevector simulation with $1024$ shots. The tested configuration is WarmStart+HGLE: SDP-relaxation warm-start~\cite{egger2021} combined with HGLE's projected surrogate objective during COBYLA optimization. Baselines from~\cite{guo2026} include Transfer (parameter transfer), TQA+Opt (trotterized quantum annealing with local optimization), Interp (interpolation), and Fourier\footnote{Fourier-ansatz parameterization~\cite{guo2026}, labeled \texttt{FourierAnsatz} in the benchmark data, distinct from the simpler Fourier initialization also present in that dataset.} (Fourier-ansatz parameterization).

\begin{table}[t]
  \centering
  \caption{Mean approximation ratio (\%) by method and graph family on the Guo~et~al.\ benchmark~\cite{guo2026}. WS+HGLE combines SDP warm-start~\cite{egger2021} with HGLE compression. Sample sizes vary by method ($48$--$72$ pairs) because not all baselines are available at all depths.}
  \label{tab:benchmark-graph}
  \renewcommand{\arraystretch}{1.15}
  \small
  \setlength{\tabcolsep}{4pt}
  \begin{tabular}{lccc}
    \toprule
    \textbf{Method} & \textbf{ER} & \textbf{Reg.} & \textbf{Overall} \\
    \midrule
    WS+HGLE     & \textbf{90.2} & \textbf{89.7} & \textbf{90.0} \\
    Transfer     & 90.1          & 87.3          & 88.7 \\
    TQA+Opt      & 85.8          & 80.9          & 84.9 \\
    Interp       & 87.1          & 83.4          & 85.3 \\
    Fourier      & 85.6          & 85.4          & 85.5 \\
    \midrule
    HGLE         & 78.8          & 82.0          & 80.4 \\
    Vanilla      & 69.6          & 76.3          & 72.9 \\
    \bottomrule
  \end{tabular}
\end{table}

\begin{figure}[t]
  \centering
  \includegraphics[width=\columnwidth]{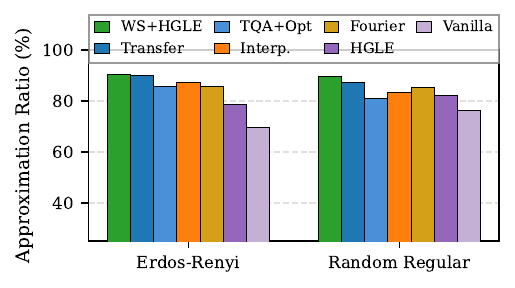}
  \caption{Approximation ratios on the Guo~et~al.\ benchmark~\cite{guo2026} by graph family ($24$ instances, $10$--$20$ qubits, depths $p{=}1,3,5$).}
  \label{fig:benchmark-guo}
\end{figure}

Table~\ref{tab:benchmark-graph} and Figure~\ref{fig:benchmark-guo} report mean approximation ratios by graph family. WarmStart+HGLE achieves the highest overall mean ($90.0\%$), exceeding Transfer ($88.7\%$), the best established baseline, consistently across both graph families. The advantage is depth-dependent: at $p{=}1$, WarmStart+HGLE leads by $+13$ percentage points ($91.6\%$ vs.\ $78.4\%$ for Interp), reflecting the initialization-sensitive regime where warm-start provides the largest benefit; at $p{=}3$ and $p{=}5$, methods with optimized parameter schedules (Fourier at $93.3\%$, TQA+Opt at $91.8\%$) match or exceed WarmStart+HGLE ($89.2\%$). Across all $72$ pairs, WarmStart+HGLE outperforms the best available baseline on $43$ ($59.7\%$; one-sided sign test $p=0.06$, not significant at the $0.05$ level).

WarmStart+HGLE is a composite: SDP warm-start~\cite{egger2021} provides initial parameters and HGLE provides the compressed surrogate objective. The benchmark does not separate the two contributions, and the concentration of the advantage at $p{=}1$ suggests that the warm-start component accounts for a substantial share. These comparisons use statevector simulation and measure end-to-end optimization quality, not HGLE compression in isolation. Historical HamLib~\cite{sawaya2024} runs on MIS (6--16 qubits, ${>}150$ instances) record larger performance gaps under unmatched conditions, motivating controlled investigation of constrained-problem regimes. A controlled ablation isolating warm-start and compression contributions is left to future work.


\section{Discussion}
\label{sec:discussion}

\noindent\textbf{What is certified.}
HGLE joins the measured Ising-feature representation with Hamiltonian-coefficient projection. The same-batch residual gives deterministic objective-distortion control, while a projector frozen from an independent pilot admits a statistically valid population-distortion interval on fresh samples. Summable failure allocation covers predictable adaptive visits, and simultaneous original-objective intervals certify a candidate ordering only under strict separation. Optional leverage-row sampling (Appendix~\ref{sec:appendix-leverage}) and region freezing (Appendix~\ref{sec:appendix-region-frozen}) are separate extensions with stronger assumptions.

\noindent\textbf{What the evidence shows.}
Across $200$ matched units, the mean retained feature-rank fraction falls from $0.931$ at $\eta=0.01$ to $0.243$ at $\eta=0.10$; complete same-batch bound satisfaction validates the deterministic implementation. In the independent pilot, population intervals show conservative coverage and nonvacuous intervals, but no condition separates nine candidates. A focused variant reducing simultaneous tests from ${\sim}576$ to $12$ certified all $114/114$ decisions, identifying multiplicity accounting as the dominant bottleneck. Same-batch certification extends to $18$--$24$ qubits in simulation and to \texttt{ibm\_marrakesh} hardware at $8$--$16$ qubits. The structured level-$1$ study preserves parameter-selection regret on its tested grids. The regional study has a decisive negative control: raw estimation at the same lower shot count behaves nearly identically (Appendix~\ref{sec:appendix-region-frozen}).

\noindent\textbf{Benchmark comparison.}
On a community Max-Cut benchmark~\cite{guo2026} (Section~\ref{sec:benchmark-comparison}), WarmStart+HGLE is competitive with established parameter-setting methods, with the advantage concentrated at low circuit depth where warm-start initialization dominates. At higher depths, optimized-parameter methods match or exceed the composite. Because the benchmark does not separate SDP warm-start~\cite{egger2021} from HGLE compression, the isolated contribution of HGLE's projected objective remains an open question. The gap between the controlled negative results and the positive low-depth benchmark observations motivates a controlled ablation of the two components and an investigation of constrained-problem regimes where the spectral-decay analysis (Section~\ref{sec:spectral-decay-results}) predicts more favorable compression.

\noindent\textbf{Scope of the result.}
The evidence does not establish practical finite-candidate decision preservation, general approximation-ratio improvement, optimizer acceleration, fewer circuit evaluations, fewer shots, lower runtime, or hardware-noise mitigation. The finite-trajectory result concerns the best visited parameter, not the global QAOA optimum.

\noindent\textbf{Computational complexity.}
HGLE adds feature-matrix construction in $O(Nd)$, randomized rank-$r$ SVD in $O(Ndr)$~\cite{martinsson2020}, and certificate evaluation in $O(d)$ per parameter evaluation. Under generic sampling, $\kappa_r \geq \sqrt{1 - r/d}$ with high probability, so achieving $\kappa_r \leq \epsilon$ requires $r \geq (1-\epsilon^2)d$. Meaningful rank reduction occurs only at moderate tolerances and on structured problems where the spectrum of $A^\top A$ decays faster than the generic flat case.

\noindent\textbf{Variance reduction.}
Coefficient projection strictly reduces estimator variance (Theorem~\ref{thm:variance-reduction}), but the bias grows with $\kappa_r$, so aggressive compression introduces bias that the conservative Cauchy--Schwarz bound does not offset (Appendix~\ref{sec:appendix-variance}). Tighter bias bounds exploiting Proposition~\ref{prop:spectral-decay-structure} may shift the balance.

\noindent\textbf{Multi-observable compression.}
The simultaneous $K$-observable certificate (Proposition~\ref{prop:multi-observable}) addresses the rank-$1$ degeneracy directly. When the application requires certified estimation of multiple observables, the shared SVD basis serves all $K$ coefficient vectors, and the rank-$1$ shortcut $\operatorname{span}\{c\}$ fails for any observable not proportional to $c$. MIS provides a natural instance with $K=2+n$ observables. The corollary that $r\geq K$ is necessary for $K$ independent observables implies that multi-observable settings generically require nontrivial compression.

\noindent\textbf{Composition with measurement-error mitigation.}
The same-batch certificate holds for any computational-basis count distribution, so it applies to counts corrected by a readout-mitigation protocol~\cite{temme2017,vandenberg2023}. The total error decomposes additively: mitigation residual plus certified compression distortion. This is validated on \texttt{ibm\_marrakesh} hardware. Secondary diagnostics (HES and CEL) are detailed in Appendix~\ref{sec:appendix-diagnostics}.


\section{Conclusion}
\label{sec:conclusion}

HGLE formulates QAOA representation compression as projection of known Hamiltonian coefficients onto a subspace learned from measured Ising features. The discarded component gives deterministic same-batch distortion control that is minimax-tight: no linear estimator on the same subspace can achieve smaller worst-case error. The residual fraction $\kappa_r^2$ admits a signal-energy interpretation governed by the spectrum of the measured feature matrix, with a fundamental lower bound $\kappa_r \geq \sqrt{1 - r/d}$ under generic sampling that explains why meaningful compression occurs only at moderate tolerances. With an independently learned pilot projector and fresh certification samples, HGLE also yields a finite-sample population-distortion interval whose validity composes over predictable adaptive visits by summable failure allocation. Simultaneous intervals preserve candidate decisions only when they are strictly separated. The certificate is complementary to classical shadows, which redesign the measurement protocol rather than the post-measurement representation, and the two can be composed.

Coefficient projection strictly reduces estimator variance (Theorem~\ref{thm:variance-reduction}), but the conservative Cauchy--Schwarz bias bound prevents a net MSE win in the tested regimes (Appendix~\ref{sec:appendix-variance}); tighter bias bounds remain an open direction. Multi-observable compression extends the certificate to $K$ observables on a shared SVD basis, defeating the rank-$1$ degeneracy: MIS experiments confirm that a shared basis certifies cost, penalty, and $n$ local observables simultaneously, while rank-$1$ projection fails for any but the full Hamiltonian. The spectral decay analysis connects compression quality to graph topology through the expected Gram matrix, with experimental confirmation that dense ER graphs compress better than sparse regular graphs at matched sizes.

The preregistered pilot demonstrates the distinction between validity and usefulness. Its held-out intervals are conservatively covering and nonvacuous, but no compressed or full-rank condition certifies a winner among nine candidates. Production scaling is therefore stopped as specified in advance. Together with the matched raw-shot control, aggressive-compression degradation, and rejected continuous-region preflight, this result rules out claims of practical decision preservation, universal acceleration, or shot savings from the present controlled evidence. On a community Max-Cut benchmark~\cite{guo2026}, WarmStart+HGLE is competitive with established parameter-setting methods at low circuit depth, though the benchmark does not separate the warm-start and compression contributions. These observations, together with the spectral-decay analysis linking compression quality to graph topology, motivate a controlled ablation of the two components and an investigation of constrained-problem regimes where HGLE's projection may provide a systematic advantage.

Future work should seek sharper valid intervals (potentially through tighter bias bounds that exploit the spectral structure of Proposition~\ref{prop:spectral-decay-structure}), less costly simultaneous procedures, or candidate designs whose objective gaps are large enough to certify without post-selection leakage. The depth-dependent spectral decay results suggest that optimized deeper circuits may enter a regime where the bias-variance tradeoff becomes favorable, but verifying this requires joint optimization of parameters and compression rank. Incremental basis construction and basis reuse matter only if the retained feature representation supports a declared downstream task and total quantum and classical costs show a genuine crossover.


\bibliographystyle{IEEEtran}
\bibliography{references}


\appendices

\section{Diagnostic Details: HES and CEL}
\label{sec:appendix-diagnostics}

The certification results in the main body do not require a predictor of whether compression will improve optimization. We nevertheless retain two sample-derived signals for configuration and monitoring; they are operational heuristics, not objective certificates.

The HGLE Effectiveness Score (HES) is a weighted geometric mean of four pilot-batch properties: energy discriminability (EDS), spectral compressibility (SCS), weight concentration (WCS), and subspace stability (SSS). A cross-sectional validation on $90$ Max-Cut instances across three graph families (Erd\H{o}s--R\'{e}nyi, random $3$-regular, planted partition) at $8$--$12$ qubits finds that SCS is a significant predictor of $\kappa_r$: pooled Spearman $\rho=-0.67$ ($p<0.001$) at rank fraction $0.25$, with the strongest per-family association on planted-partition graphs ($\rho=-0.75$). EDS and WCS are not consistent predictors across families.

Compression Error over the Landscape (CEL) records the normalized disagreement between the projected-coefficient estimate and the raw sample mean at each evaluation. A trajectory study on $10$ Erd\H{o}s--R\'{e}nyi instances ($n=10$, depth-$2$ QAOA, $30$ evaluations each) shows that the deterministic certificate bound $\kappa_r\norm{c}_2\norm{\bar\vphi_t}_2$ holds at all $300$ evaluations, and CEL tracks the certificate with mean Spearman $\rho=+0.57$ (significant in $9/10$ instances). CEL thus provides a cheap monitoring signal that correlates with the theoretical certificate along optimization trajectories, without requiring additional circuit execution.

HES and CEL remain secondary diagnostics; the coefficient residual $\kappa_r$ is the quantity used by the certification results in this paper. The remainder of this appendix gives the full definitions.

\subsection{HES Definition and Components}

The HGLE Effectiveness Score (HES) summarizes four properties of a pilot batch. For clipped component values $\bar s_k=\max(s_k,10^{-3})$, the composite is the weighted geometric mean
\begin{equation}
 \HES=\left(\prod_{k=1}^{4}\bar s_k^{w_k}\right)^{1/\sum w_k},
 \quad w=(1.5,1.5,1,1).
 \label{eq:hes-appendix}
\end{equation}
The four components are:
\begin{itemize}
\item \textbf{Energy Discriminability Score (EDS):} Measures how well-separated the sampled energy values are, using the ratio of the energy range to its standard deviation. Higher EDS indicates that low-energy weighting can effectively distinguish good from poor samples.
\item \textbf{Spectral Compressibility Score (SCS):} Uses the participation ratio of the feature matrix's singular-value spectrum to measure effective rank relative to the ambient dimension. Higher SCS indicates faster spectral decay and better low-rank approximability.
\item \textbf{Weight Concentration Score (WCS):} Quantifies whether Boltzmann weighting produces a useful distribution via the effective sample size ratio. At default $\tau=1$ with count-weighted production objectives, WCS describes the pilot distribution, not the production weighting.
\item \textbf{Subspace Stability Score (SSS):} Bootstrap resampling of the feature matrix measures how much the leading right-singular subspace changes under resampling. Higher SSS indicates a more stable basis.
\end{itemize}

A one-shot policy may map the composite score and its weakest component to a rank fraction and production shot count, then freeze that configuration until an explicit reset. Its thresholds were empirically calibrated and do not follow from the coefficient-residual bounds.

\subsection{Per-Family HES Channel Correlations}

Tests of the intended channels are topology-dependent. On $157$ planted-partition instances, SCS has the predicted residual association ($\rho=-0.28$, $p=0.0004$), while EDS and WCS are unsupported. On $160$ random $3$-regular instances, EDS is supported, WCS is unsupported, and SSS has the wrong sign. The validation described at the beginning of this appendix confirms the SCS--$\kappa_r$ association on a broader set of instances and graph families.

Since energy scaling changes concentration at fixed $\tau$, these results neither calibrate WCS across Hamiltonian scales nor establish optimality of $\tau=1$. Future studies should normalize energy, report effective sample size, and test temperature sensitivity.

\subsection{CEL Definition and Tracking}

Compression Error over the Landscape (CEL) records normalized disagreement between the projected-coefficient estimate and the raw sample mean from the same counts:
\begin{equation}
 \mathrm{CEL}_t=
 \frac{\abs{\Fhat_r(\vthet_t)-\Fhat_{\mathrm{raw}}(\vthet_t)}}
 {\max\!\left\{\operatorname{std}(\{E_j(\vthet_t)\}),
 \epsilon_{\mathrm{CEL}}\right\}}.
 \label{eq:cel-appendix}
\end{equation}
Recording CEL requires no additional circuit execution when both estimates are available. An exponential moving average may be compared with an early-run baseline to flag drift, and an opt-in controller may respond by restarting from a fresh pilot. Monitoring alone does not alter the objective or optimizer.

CEL is neither ground-truth population error nor a guarantee that restarting improves the final solution. Its smoothing coefficient, baseline length, trigger ratio, and minimum-call guard require calibration. The available configured-system comparisons inherit the uncalibrated $\tau=1$ WCS pilot and change monitoring and optimizer settings together, so they neither establish temperature robustness nor isolate a causal benefit from CEL-triggered retuning.

\subsection{Implementation Notes}

The one-shot adaptive procedure obtains HES from a $512$-shot pilot. Its WCS component uses default $\tau=1$, which can affect regime selection. The selected regime fixes rank fraction and sample count until reset; production compression remains count-weighted, while its nominal $\tau$ affects only optional diagnostics. The pilot is tracked separately from subsequent objective evaluations and should be included when reporting the total quantum sampling cost.

The practical HES calculation requires the sample energies, counts, and singular spectrum of the weighted feature matrix. Subspace stability additionally performs bootstrap resampling, so diagnostic computation adds classical SVD work beyond the energy estimate itself. This overhead is paid once per pilot unless the configuration is explicitly reset.

One CEL value is recorded for each real objective evaluation; cached accesses are not new measurements. The running state comprises the complete history, an exponential moving average, a baseline fixed after the configured number of early calls, and an optional trigger flag. For chunked retuning, each restarted chunk receives fresh estimator and optimizer state and begins from the best parameters observed so far.

\FloatBarrier

\section{Optional Leverage-Score Row Sampling}
\label{sec:appendix-leverage}

This appendix presents the optional leverage-score row-sampling extension referenced in Section~\ref{sec:hgle}. The experiments in the main body do not use this mechanism; it is included for completeness.

\subsection{Row Sampling Construction}

Let $A = U\Sigma V^\top$ be the thin SVD of $A$ and fix target rank $r\leq
\operatorname{rank}(A)$.  The rank-$r$ truncated SVD is $A_r = U_r\Sigma_r
V_r^\top$.  Define \emph{row leverage scores}:
\begin{equation}
  \ell_t = \norm{e_t^\top U_r}^2_2, \quad p_t = \ell_t / r,
  \quad t=1,\ldots,N.
  \label{eq:lev}
\end{equation}
We construct a sampling matrix $S\in\R^{m\times N}$ by drawing $m$ row indices $\xi_1,\ldots,\xi_m$ independently according to $\{p_t\}$, with rescaling $S_{j,t} = 1/\sqrt{m p_t}$ when $t=\xi_j$ (and zero otherwise).  The \emph{compressed matrix} is
\begin{equation}
  \Atil = SA \in \R^{m\times d}, \quad m \ll N.
  \label{eq:Atil}
\end{equation}

\subsection{Formal Guarantees}

\noindent\textbf{Scope and assumptions.}
The first result below is a standard randomized-linear-algebra guarantee for row sampling from the exact top-$r$ leverage distribution. It preserves the selected column space but does not certify an objective. The second result is the HGLE-specific composition: conditional on the realized embedding and a measured tail-separation condition, it combines projector perturbation with the Hamiltonian coefficient residual to certify the sketched projected objective. Neither result reduces quantum shots, and the current exact-leverage construction performs a full SVD before sketching.

\begin{definition}[$\varepsilon$-subspace embedding]
\label{def:subspace-embedding}
$S$ is an $\varepsilon$-subspace embedding for $\mathrm{col}(A_r)$ if
\[
  (1-\varepsilon)\norm{U_r x}^2_2 \;\leq\; \norm{SU_r x}^2_2 \;\leq\;
  (1+\varepsilon)\norm{U_r x}^2_2 \quad \forall x\in\R^r.
\]
\end{definition}

\begin{proposition}[Rank preservation]
If $S$ is an $\varepsilon$-subspace embedding for $\mathrm{col}(A_r)$ with
$0<\varepsilon<1$, then $\operatorname{rank}(SA_r)=r$.
\end{proposition}
\noindent\textit{Proof.}
This is standard; we include it for completeness.
For any nonzero $x\in\R^r$, the embedding condition gives $\norm{SU_r x}^2_2 \geq (1-\varepsilon)\norm{U_r x}^2_2 > 0$. Hence $SU_r$ has trivial null space, so $\operatorname{rank}(SU_r)=r$. Since $\Sigma_r$ is invertible, $SA_r = (SU_r)\Sigma_r V_r^\top$ has rank $r$ as well. \hfill$\square$

\begin{theorem}[Leverage-score guarantee~\cite{drineas2012}]
\label{thm:lev}
The following is due to Drineas et al.~\cite{drineas2012}; we restate it in our notation.
For $0<\varepsilon,\delta<1$ there exists an absolute constant $C>0$ such
that, if
\begin{equation}
  m \geq C\,\frac{r\log(r/\delta)}{\varepsilon^2},
  \label{eq:m}
\end{equation}
then with probability at least $1-\delta$, $S$ is an $\varepsilon$-subspace embedding for $\mathrm{col}(A_r)$, and consequently $\operatorname{rank}(\Atil_r)=r$.
\end{theorem}

Theorem~\ref{thm:lev} preserves the selected rank-$r$ column space of $A_r$ under row sampling. For an approximately rank-$r$ matrix, it does not by itself show that the top right-singular projector of $SA$ approximates $P_r$, nor does it certify an energy estimator. The residual certificates in Section~\ref{sec:certified-rank} therefore use the right-singular projector of the observed feature matrix directly and do not require row sampling.

A separate, tail-aware statement certifies a projected estimator built from the sketch. Write $A=A_r+E$, let $\widetilde P_r=\widetilde V_r\widetilde V_r^\top$ be the top-$r$ right-singular projector of $SA$, and define
\begin{equation}
 \alpha_S=\sqrt{\lambda_{\min}(U_r^\top S^\top S U_r)},
 \qquad \tau_S=\norm{SE}_{\rm op}.
 \label{eq:sketch-diagnostics}
\end{equation}
Both quantities are measured for the realized sketch.

\begin{theorem}[Tail-aware sketched projection]
\label{thm:sketched-projection}
Suppose the realized $S$ satisfies the embedding event in Definition~\ref{def:subspace-embedding} and
$\tau_S<\alpha_S\sigma_r(A)$. Then
\begin{equation}
 \norm{\widetilde P_r-P_r}_{\rm op}
 \leq \frac{\tau_S}{\alpha_S\sigma_r(A)-\tau_S},
 \label{eq:sketch-projector-bound}
\end{equation}
and, for any feature mean $\bar\vphi$,
\begin{equation}
 \abs{c^\top\bar\vphi-(\widetilde P_rc)^\top\bar\vphi}
 \leq
 \left[\kappa_r+
 \frac{\tau_S}{\alpha_S\sigma_r(A)-\tau_S}\right]
 \norm{c}_2\norm{\bar\vphi}_2.
 \label{eq:sketch-objective-bound}
\end{equation}
If the sample size obeys~\eqref{eq:m}, the embedding event additionally has probability at least $1-\delta$; this probability statement and the measured deterministic certificate are reported separately.
\end{theorem}
\begin{proof}
The embedding event gives
$\sigma_r(SA_r)=\sigma_r(SU_r\Sigma_rV_r^\top)\geq\alpha_S\sigma_r(A)$.
Treating $SE$ as a perturbation and using Wedin's $\sin\Theta$ bound~\cite{horn2013} under the stated separation condition gives~\eqref{eq:sketch-projector-bound}. Next,
\[
 \norm{(I-\widetilde P_r)c}_2
 \leq \norm{(I-P_r)c}_2+
 \norm{(P_r-\widetilde P_r)c}_2.
\]
Cauchy--Schwarz and~\eqref{eq:sketch-projector-bound} yield~\eqref{eq:sketch-objective-bound}.
\end{proof}

\begin{corollary}[Exact-rank sketch]
If $A=A_r$ and the embedding event holds, then $\tau_S=0$ and
$\widetilde P_r=P_r$. Thus the sketch introduces no projector error and
Eq.~\eqref{eq:sketch-objective-bound} reduces to the coefficient-residual certificate. In particular, it is exact whenever the evaluated feature mean lies in $\operatorname{row}(A_r)$.
\end{corollary}

This mechanism is opt-in and distinct from the legacy leverage-weighted energy average, which is not covered by Theorem~\ref{thm:sketched-projection}. The current implementation computes a full SVD to obtain exact leverage probabilities before sampling, so it establishes neither a basis-construction speedup nor a reduction in quantum shots. Its benefit is a separately diagnosed classical row sketch when the measured embedding and tail conditions hold.

\subsection{Leverage Row-Sketch Algorithm}
\begin{algorithm}[H]
\caption{Optional leverage row sketch within HGLE}
\label{alg:hgle}
\begin{algorithmic}[1]
\Require Ising Hamiltonian $\HC$,
         parameters $\vthet$,
         shots $N$,
         rank $r$,
         distortion $\varepsilon$,
         failure level $\delta$
\Ensure Compressed matrix $\Atil = SA \in \R^{m\times d}$
\State Prepare $\ket{\psi(\vthet)}$ and measure $N$ times:
       $s^{(1)},\ldots,s^{(N)}\in\{-1,+1\}^n$
\State Compute energies $E^{(t)} = E(s^{(t)})$
\State Choose weights $w_t$ via \eqref{eq:wunif} or \eqref{eq:wenergy}
\State Build weighted feature matrix $A$ via \eqref{eq:A}
\State Estimate effective rank $r$ and approximate leading basis $U_r$
\State Compute leverage scores $\ell_t = \norm{e_t^\top U_r}^2_2$,
       set $p_t = \ell_t/r$
\State Choose $m$ per \eqref{eq:m}; sample indices $\xi_1,\ldots,\xi_m\sim p_t$
\State Form $S$ and return $\Atil = SA$
\end{algorithmic}
\end{algorithm}

\FloatBarrier

\section{Region-Frozen HGLE}
\label{sec:appendix-region-frozen}

This appendix presents the region-frozen variant and its experimental results, referenced from the main body.

The same-batch certificate in Section~\ref{sec:certified-rank} does not control derivatives when the projector changes with \(\vthet\). We therefore consider an opt-in variant in which a projector \(P_r\), constructed from a pilot at the center of a declared region \(D\subset\R^{d_\theta}\), remains fixed throughout the local search. A trust-region optimizer is operationally convenient here because its radius supplies the declared locality: trial points can be constrained to \(D\), and leaving \(D\) triggers rejection or explicit recertification rather than extrapolation of the frozen basis. The regional theorem depends on the fixed projector and certified geometry, not on trust-region convergence theory; another optimizer that enforces the same domain restriction could be used.

\noindent\textbf{Scope and assumptions.}
Conditional on the realized pilot batch and projector, the regional statements are deterministic. Uniform derivative control requires a fixed projector and sufficient differentiability on an open neighborhood of the declared region. Minimizer retention additionally requires a compact convex region with sufficiently smooth boundary, positive outward flow, strong convexity, and residual small enough to preserve both margins. The finite-cover result is rigorous only with proved cover radii, certified derivative envelopes, and positive corrected margins. It is not inferred from a visually dense numerical grid.

Let
\begin{equation}
 F(\vthet)=c^\top\mu(\vthet),\quad
 F_r(\vthet)=(P_rc)^\top\mu(\vthet),\quad
 e_r=(I-P_r)c,
 \label{eq:regional-objectives}
\end{equation}
where \(\mu(\vthet)=\E_{s\sim p_{\vthet}}[\vphi(s)]\). Define
\begin{align}
 B_0&=\sup_{\vthet\in D}\norm{\mu(\vthet)}_2, &
 B_1&=\sup_{\vthet\in D}\norm{J_\mu(\vthet)}_{\rm op},\nonumber\\
 B_2&=\sup_{\substack{\vthet\in D\\\norm{a}_2=1}}
 \norm{\nabla^2(a^\top\mu(\vthet))}_{\rm op}.
 \label{eq:regional-constants}
\end{align}

\begin{theorem}[Uniform regional control]
\label{thm:regional-control}
If \(\mu\) is twice continuously differentiable on an open neighborhood of \(D\), then
\begin{align}
 \sup_{\vthet\in D}\abs{F-F_r}&\le B_0\norm{e_r}_2,\nonumber\\
 \sup_{\vthet\in D}\norm{\nabla F-\nabla F_r}_2
 &\le B_1\norm{e_r}_2,\label{eq:regional-bounds}\\
 \sup_{\vthet\in D}\norm{\nabla^2F-\nabla^2F_r}_{\rm op}
 &\le B_2\norm{e_r}_2.\nonumber
\end{align}
\end{theorem}
\begin{proof}
Each difference is the corresponding derivative of \(e_r^\top\mu(\vthet)\). The bounds follow from Cauchy--Schwarz and the definitions in~\eqref{eq:regional-constants}.
\end{proof}

\begin{corollary}[Local minimizer retention]
\label{cor:regional-minimizer}
Suppose \(D\) is compact and convex with a sufficiently smooth boundary and outward unit normal \(n(\vthet)\), \(F\) is \(\lambda\)-strongly convex on \(D\), and
\begin{equation}
 b=\inf_{\vthet\in\partial D}
 \langle\nabla F(\vthet),n(\vthet)\rangle>0.
 \label{eq:regional-boundary-margin}
\end{equation}
The boundary condition places the unique minimizer \(\vthet^\star\) of \(F\) in the interior, where \(\nabla F(\vthet^\star)=0\). If \(B_1\norm{e_r}_2<b\) and \(B_2\norm{e_r}_2<\lambda\), then \(F_r\) retains inward descent at the boundary, is \((\lambda-B_2\norm{e_r}_2)\)-strongly convex, and has a unique interior minimizer \(\vthet_r^\star\). Moreover,
\begin{equation}
 \norm{\vthet_r^\star-\vthet^\star}_2
 \le \frac{B_1\norm{e_r}_2}
 {\lambda-B_2\norm{e_r}_2}.
 \label{eq:regional-displacement}
\end{equation}
\end{corollary}
\begin{proof}
On the boundary, Theorem~\ref{thm:regional-control} and Cauchy--Schwarz give
\[
 \langle\nabla F_r,n\rangle
 \ge b-\norm{\nabla F_r-\nabla F}_2
 \ge b-B_1\norm{e_r}_2>0,
\]
so a constrained minimizer of \(F_r\) cannot lie on the boundary. Weyl's inequality~\cite{horn2013} gives
\[
 \lambda_{\min}(\nabla^2F_r)
 \ge \lambda-B_2\norm{e_r}_2>0,
\]
hence \(F_r\) is strongly convex and has a unique interior minimizer. Since \(\nabla F(\vthet^\star)=0\),
\(\norm{\nabla F_r(\vthet^\star)}_2\le B_1\norm{e_r}_2\). Strong monotonicity of \(\nabla F_r\), together with \(\nabla F_r(\vthet_r^\star)=0\), yields
\[
 (\lambda-B_2\norm{e_r}_2)
 \norm{\vthet_r^\star-\vthet^\star}_2
 \le B_1\norm{e_r}_2,
\]
which proves~\eqref{eq:regional-displacement}.
\end{proof}

\subsection{From Finite Covers to Continuous Bounds}
The preceding result requires continuous-domain constants. They can be obtained rigorously from finite covers only when the cover radii and derivative envelopes are themselves proved. For a Euclidean ball $D$ of radius $R$, let $G\subset D$ and $G_\partial\subset\partial D$ have proved covering radii $h$ and $h_\partial$. Suppose exact values or rigorous enclosures on these covers give grid extrema $B_{0,G}$, $B_{1,G}$, $B_{2,G}$, $\lambda_G$, and $b_G$, matching the quantities in~\eqref{eq:regional-constants}, the minimum objective-Hessian eigenvalue, and the outward-gradient margin. Let $L_2$ and $L_3$ be certified uniform operator-norm envelopes for the second and third derivatives of $\mu$ on an open neighborhood of $D$.

\begin{theorem}[Rigorous grid-to-continuum correction]
\label{thm:grid-continuum}
Under the assumptions above, define
\begin{align}
 \widehat B_1&=B_{1,G}+L_2h,&
 \widehat B_0&=B_{0,G}+\widehat B_1h,\nonumber\\
 \widehat B_2&=B_{2,G}+L_3h,&
 \widehat\lambda&=\lambda_G-\norm{c}_2L_3h,
 \label{eq:continuum-interior}
\end{align}
and
\begin{equation}
 \widehat b=b_G-\norm{c}_2
 \left(\widehat B_2+\frac{\widehat B_1}{R}\right)h_\partial.
 \label{eq:continuum-boundary}
\end{equation}
Then $B_i\leq\widehat B_i$ for $i=0,1,2$, $\nabla^2F(\vthet)\succeq\widehat\lambda I$ throughout $D$, and the boundary margin of $F$ is at least $\widehat b$. If $\widehat\lambda>0$ and $\widehat b>0$, $F$ is $\widehat\lambda$-strongly convex and Corollary~\ref{cor:regional-minimizer} is rigorous with these corrected constants.
\end{theorem}
\begin{proof}
For any point in $D$, choose a cover point within $h$. The derivative envelopes and the fundamental theorem of calculus along the connecting segment give the Hessian, Jacobian, and feature-mean corrections in~\eqref{eq:continuum-interior}; Weyl's inequality~\cite{horn2013} gives the curvature correction. On the sphere, the outward normal is $n(\vthet)=(\vthet-\vthet_0)/R$ and is $1/R$-Lipschitz. Adding the change in the objective gradient and the change in the normal gives~\eqref{eq:continuum-boundary}.
\end{proof}

\begin{table}[!htbp]
\centering
\caption{Rigorous continuum preflight on a representative $n=6$, degree-$2$, unit-weight instance with radius $0.12$. Positive raw grid extrema do not survive the proved covering correction.}
\label{tab:continuum-preflight}
\begin{tabular}{@{}lc@{}}
\toprule
Quantity & Verified value or outcome \\
\midrule
Second-derivative envelope $L_2$ & $94.49$ \\
Third-derivative envelope $L_3$ & $394.85$ \\
Interior and boundary radii $(h,h_\partial)$ & $(0.0385,0.0235)$ \\
Grid curvature minimum $\lambda_G$ & $5.322$ \\
Grid outward-flow minimum $b_G$ & $0.713$ \\
Corrected curvature $\widehat\lambda$ & nonpositive \\
Continuous regional certificate & rejected \\
\bottomrule
\end{tabular}
\end{table}

Table~\ref{tab:continuum-preflight} makes the decisive negative result explicit: the raw cover diagnostics appear favorable, but the curvature correction makes $\widehat\lambda$ nonpositive. The procedure therefore issues no real-instance continuous-domain certificate. This rejection demonstrates that the certificate is falsifiable rather than an automatic relabeling of grid evidence.

\subsection{Region-Frozen Experimental Results}

We test the opt-in construction on level-$1$ weighted Max-Cut with $6$ and $8$ vertices, degrees $2$ and $3$, and unit or bounded integer weights. Of $48$ candidate units, $39$ admit a genuinely compressed empirical-grid certificate; the remaining nine are rejected rather than forced into an uncertified rank. The accepted units comprise $14$ complete graph clusters, three initializations per unit, a radius-$0.12$ region, a $60$-evaluation cap, and a $1024$-shot regional pilot. Exact original-Hamiltonian expectations are used only for post-hoc scoring. Uncertainty uses $5000$ deterministic bootstrap replicates over complete graph clusters.

\begin{table}[!htbp]
\centering
\caption{Paired 512-shot contrasts against raw estimation at 1,024 shots. Positive normalized differences favor the reference. Brackets are 95\% bootstrap intervals.}
\label{tab:region-frozen-shots}
\scriptsize
\setlength{\tabcolsep}{1.5pt}
\begin{tabular}{@{}lcc@{}}
\toprule
\textbf{Arm} & \textbf{$\Delta$ best excess} & \textbf{$\Delta$ AUC} \\
\midrule
Raw
  & $.00078\;[{-}.00073,\,.00396]$
  & $.00057\;[{-}.00076,\,.00321]$ \\[2pt]
Dyn.\ cert.
  & $.00083\;[{-}.00068,\,.00399]$
  & $.00059\;[{-}.00075,\,.00321]$ \\[2pt]
Reg.-frozen$^{\dagger}$
  & $.00072\;[{-}.00076,\,.00388]$
  & $.00054\;[{-}.00078,\,.00315]$ \\
\bottomrule
\multicolumn{3}{@{}l@{}}{\scriptsize $^{\dagger}$Includes 1,024-shot pilot amortized over three starts.} \\
\end{tabular}
\end{table}

Table~\ref{tab:region-frozen-shots} compares each $512$-shot arm with paired raw estimation at $1024$ shots. Region-frozen HGLE has best-visited excess difference $0.00072$ with interval $[-0.00076,0.00388]$ and cap-padded AUC difference $0.00054$ with interval $[-0.00078,0.00315]$. Its amortized pilot-inclusive cost is $31{,}061$ shots versus $61{,}440$ production shots for the reference. Both upper bounds fall below an inherited $0.005$ sensitivity threshold, but that threshold was not preregistered for this sweep.

\noindent\textbf{Decisive matched control.}
Raw estimation at $512$ shots has best-visited difference $0.00078$ and AUC difference $0.00057$, nearly identical to region-frozen HGLE. Selected-point and optimizer-returned quality also fail the same sensitivity check. The lower-shot preservation therefore cannot be attributed uniquely to region freezing. The supported conclusion is narrower: accepted local regimes preserve best-visited quality and AUC under the inherited threshold, while region-frozen runs at $128$ and $256$ shots do not. This control rules out a general HGLE shot-saving claim.

\FloatBarrier

\section{Evaluation Considerations and Trust-Region Loop}
\label{sec:appendix-evaluation}

This appendix collects the evaluation considerations from the main body and the trust-region classical loop algorithm.

\subsection{Implementation Details}

\noindent\textbf{Approximate leverage scores.}
Computing the exact SVD of $A$ costs $O(Nd^2)$, which is expensive for large $N$. Approximate $U_r$ and leverage scores can instead be obtained with a randomized range finder or sparse initial sketch, and approximate scores can support standard embedding guarantees~\cite{drineas2012}. The evaluated exact-leverage path computes the full SVD first, so the present evidence establishes neither lower basis-construction cost nor an end-to-end speedup.

\noindent\textbf{Weighting strategy.}
All primary certification experiments use empirical count weighting, so Boltzmann $\tau$ is inactive in their objectives. Low-energy weights~\eqref{eq:wenergy} define a different target and require a declared $\tau$, a matched raw control, and sensitivity analysis; that study is outside the present evidence.

\noindent\textbf{Rank, independence, and multiplicity.}
For fixed-rank HGLE, the effective rank $r$ can be selected as the number of singular values exceeding a user-specified spectral threshold, and the embedding dimension $m$ is set per~\eqref{eq:m}. The same-batch certified variant instead selects the smallest rank satisfying the observed coefficient-residual tolerance at each call; this is deterministic conditional on that batch. Population certification has a stricter contract: pilot observations construct and freeze the projector, while disjoint certification observations estimate projected and residual means. Any inspection across ranks, candidates, basis methods, certificate methods, or the two interval components consumes the simultaneous failure budget. Table~\ref{tab:certified-controlled} shows that same-batch retained fraction varies with tolerance and problem size. An optional sparse oblivious stage (e.g., CountSketch) can further compress $\Atil$ for streaming applications.

\noindent\textbf{Regional-certificate provenance.}
Every continuum-certificate observation record is bound by a cryptographic fingerprint to its exact interior and boundary covers. Factory-constructed covers are trusted by default; externally supplied covers require an explicit trust decision after their claimed radii have been independently proved. The procedure records the provenance of exact or enclosed grid derivatives and derivative envelopes, and rejects nonpositive corrected margins.

\noindent\textbf{Pilot diagnostics.}
The one-shot adaptive procedure obtains HES from a $512$-shot pilot and fixes rank fraction and sample count until reset. The region-frozen variant instead uses a pilot to certify and freeze one projector over a declared local region; probes outside that region are rejected, and recertification is explicit. The pilot is tracked separately from subsequent objective evaluations and should be included when reporting the total quantum sampling cost. Full diagnostic definitions and implementation details are in Appendix~\ref{sec:appendix-diagnostics}.

\begin{table}[!htbp]
\centering
\caption{Mechanisms evaluated in the HGLE framework and their analytical roles.}
\label{tab:mechanism-evaluation}
\small
\setlength{\tabcolsep}{4pt}
\begin{tabular}{p{0.32\columnwidth}p{0.58\columnwidth}}
\toprule
Mechanism & Evaluation role \\
\midrule
Leverage HGLE & Tests fixed-rank Hamiltonian-aware objective compression \\
Pilot-frozen certification & Gives valid held-out population intervals; pilot certifies no candidate decision \\
Practical HES & Assesses whether pilot-batch diagnostics characterize favorable compression regimes \\
One-shot adaptation & Evaluates configuration selection from a fixed pilot diagnosis \\
Region-frozen HGLE & Tests an opt-in fixed projector with conditional regional control \\
CEL tracking & Measures compressed/raw objective disagreement along visited trajectories \\
CEL retuning & Tests whether threshold-triggered reconfiguration improves subsequent search \\
HES correspondence & Examines how diagnostic channels depend on graph family and structure \\
\bottomrule
\end{tabular}
\end{table}

\subsection{Trust-Region Classical Loop}

\begin{algorithm}[!htbp]
\caption{Reduced Classical Loop for Estimating $(\vgam,\vbet)$}
\label{alg:loop}
\begin{algorithmic}[1]
\Require Initial $\vthet_0$ (random or grid-sampled),
         trust-region radii $\{\rho_k\}$,
         probe displacements $\{\Delta^{(j)}\}_{j=0}^q$
\Ensure Estimated parameter vector $\hat\vthet$
\For{$k = 0, 1, 2, \ldots$ until convergence}
  \For{each probe $j = 0,\ldots,q$}
    \State Set $\vthet_k^{(j)} = \vthet_k + \Delta^{(j)}$
    \State Run HGLE at $\vthet_k^{(j)}$, obtain $\Atil^{(j)}$
    \State Compute reduced basis $\Vtil_r^{(j)}$ from $\Atil^{(j)}$
    \State Estimate reduced objective $\Fhat_r(\vthet_k^{(j)})$
  \EndFor
  \State Fit local quadratic:
         $m_k(\Delta)=a_k+g_k^\top\Delta+\tfrac{1}{2}\Delta^\top B_k\Delta$
  \State Trust-region:
         $\Delta^*_k = \arg\min_{\norm{\Delta}_2\leq\rho_k} m_k(\Delta)$
  \State Set $\vthet_{k+1} = \vthet_k + \Delta^*_k$
  \If{$\norm{\Delta^*_k}_2$ and objective below tolerance}
    \State \textbf{break}
  \EndIf
\EndFor
\State \Return $\hat\vthet = \vthet_k$
\end{algorithmic}
\end{algorithm}

We use a derivative-free trust-region framework~\cite{conn2000} for two distinct reasons. First, finite-shot QAOA supplies noisy function values but not reliable gradients, while a local quadratic model can still propose curvature-informed steps and restrict them to a neighborhood in which that model is considered credible. Second, this locality provides a natural enforcement mechanism for the optional region-frozen construction (Appendix~\ref{sec:appendix-region-frozen}): probes remain inside the declared certified region, and a step that would leave it must be rejected or preceded by explicit recertification. This compatibility motivates the choice, but it is not a theorem that HGLE converges faster with trust-region methods; the optimizer-factorial control in Section~\ref{sec:controlled-certified} detects no favorable HGLE--optimizer interaction.

The initial parameter vector $\vthet_0$ may be drawn randomly from a bounded domain or selected via coarse grid search over the HGLE projected objective. The local quadratic is defined over the QAOA parameter vector $\vthet$. For $p_\theta=\dim(\vthet)$ variational parameters, a full quadratic model contains $O(p_\theta^2)$ coefficients and therefore requires a correspondingly sized interpolation set. Unlike simplex methods~\cite{powell1994,nelder1965}, the quadratic model captures curvature; unlike gradient-based methods~\cite{liu1989}, it does not require gradient estimates. Its adaptive radius $\rho_k$ limits extrapolation from locally sampled objective values.

\FloatBarrier

\section{Variance Reduction Experiment}
\label{sec:appendix-variance}

Theorem~\ref{thm:variance-reduction} predicts that the projected estimator $\Fhat_r=(P_rc)^\top\bar\vphi$ has lower variance than the raw estimator $\Fhat=c^\top\bar\vphi$ when the projector $P_r$ is fixed independently of the test data. We validate this with a pilot/test split: $1024$ pilot shots provide the SVD basis, and $4096$ test shots provide independent samples. Variance is estimated from $500$ bootstrap resamples of the test data, using bootstrap batches of $256$ samples each. The experiment covers three graph families at $n\in\{8,10,12\}$ with $10$ instances per cell.

\begin{table}[!htbp]
\centering
\caption{Mean variance reduction (\%) and $\kappa_r$ across graph families. Reduction is $100\times(\text{Var}(\Fhat)-\text{Var}(\Fhat_r))/\text{Var}(\Fhat)$.}
\label{tab:variance-reduction}
\small
\setlength{\tabcolsep}{3pt}
\begin{tabular}{llccc}
\toprule
Family & & $r/d=0.10$ & $r/d=0.25$ & $r/d=0.50$ \\
\midrule
ER  & Reduction & $3.7\%$ & $0.7\%$ & $0.4\%$ \\
    & $\kappa_r$ & $0.245$ & $0.150$ & $0.087$ \\
\midrule
Regular & Reduction & $5.6\%$ & $3.0\%$ & $1.3\%$ \\
        & $\kappa_r$ & $0.337$ & $0.229$ & $0.131$ \\
\midrule
Planted & Reduction & $3.4\%$ & $1.7\%$ & $0.5\%$ \\
        & $\kappa_r$ & $0.246$ & $0.167$ & $0.099$ \\
\bottomrule
\end{tabular}
\end{table}

The variance reduction is consistently positive and strongest at low rank fractions, where the projection discards the most noisy directions. Three-regular graphs show the largest reduction ($5.6\%$ at $r/d=0.10$), consistent with their higher $\kappa_r$ values: when the coefficient vector has a larger component in the discarded subspace, more variance is removed by projection. However, the magnitude of variance reduction is modest (below $6\%$ across all conditions), and the Cauchy--Schwarz bias bound $\kappa_r\norm{c}\norm{\bar\vphi}$ is conservative enough that the MSE of the projected estimator exceeds that of the raw estimator in all tested conditions. This demonstrates the bias-variance tradeoff predicted by Corollary~\ref{cor:bias-variance}: variance reduction is a structural property of coefficient projection, but extracting a net MSE benefit requires either tighter bias bounds or regimes where the bias is inherently small (i.e., fast spectral decay giving small $\kappa_r$ at low rank).

\end{document}